\documentclass{article}

\PassOptionsToPackage{numbers,sort&compress}{natbib}

\usepackage[final]{neurips_2025}

\usepackage[numbers]{natbib}
\usepackage[utf8]{inputenc} 
\usepackage[T1]{fontenc}    
\usepackage{hyperref}       
\usepackage{url}            
\usepackage{booktabs}       
\usepackage{amsfonts}       
\usepackage{nicefrac}       
\usepackage{microtype}      
\usepackage{xcolor}         

\usepackage{wrapfig}

\usepackage{caption}

\usepackage{float}
\floatstyle{plaintop}
\restylefloat{table}
\usepackage{graphicx}

\usepackage{booktabs} 
\usepackage{enumitem}
\usepackage{mhchem}
\usepackage{algorithm}
\usepackage{algorithmic}

\usepackage{amsmath}
\usepackage{amssymb}
\usepackage{mathtools}
\usepackage{amsthm}
\usepackage{bm}
\theoremstyle{plain}
\newtheorem{theorem}{Theorem}[section]

\newtheorem{lemma}[theorem]{Lemma}
\newtheorem{corollary}[theorem]{Corollary}
\theoremstyle{definition}

\theoremstyle{remark}

\DeclareMathOperator*{\argmin}{arg\,min}

\usepackage[capitalize,noabbrev]{cleveref}

\title{PALQO: Physics-informed Model for Accelerating Large-scale Quantum Optimization}

\author{%
  Yiming Huang$^{1}$\thanks{Equal contribution.} , Yajie Hao$^{2*}$, Jing Zhou$^5$, Xiao Yuan$^{1}$\thanks{Corresponding authors.} , Xiaoting Wang$^{2\dagger}$, Yuxuan Du$^{3,4\dagger}$\\
  $^1$Center on Frontiers of Computing Studies, \\
  and School of Computer Science, Peking University, Beijing, China\\
  $^2$Institute of Fundamental and Frontier Sciences\\University of Electronic Science and Technology of China, Chengdu, China\\
  $^3$College of Computing and Data Science\\
  Nanyang Technological University, Singapore, Singapore\\
  $^4$School of Physical and Mathematical Sciences\\
  Nanyang Technological University, Singapore, Singapore\\
  $^5$Department of Physics, Fudan University, Shanghai, China\\
  \texttt{yiminghwang@gmail.com, yuxuan.du@ntu.edu.sg,  xiaoyuan@pku.edu.cn}\\
  \texttt{xiaoting@uestc.edu.cn}
}
\begin{document}

\maketitle

\begin{abstract}
Variational quantum algorithms (VQAs) are leading strategies to reach practical utilities of near-term quantum devices. However, the no-cloning theorem in quantum mechanics precludes standard backpropagation, leading to prohibitive quantum resource costs when applying VQAs to large-scale tasks. To address this challenge, we reformulate the training dynamics of VQAs as a nonlinear partial differential equation and propose a novel protocol that leverages physics-informed neural networks (PINNs) to model this dynamical system efficiently. Given a small amount of training trajectory data collected from quantum devices, our protocol predicts the parameter updates of VQAs over multiple iterations on the classical side, dramatically reducing quantum resource costs. Through systematic numerical experiments, we demonstrate that our method achieves up to a 30x speedup compared to conventional methods and reduces quantum resource costs by as much as 90\% for tasks involving up to 40 qubits, including ground state preparation of different quantum systems, while maintaining competitive accuracy. Our approach complements existing techniques aimed at improving the efficiency of VQAs and further strengthens their potential for practical applications.
\end{abstract}

\section{Introduction}\label{Introduction}
Modern quantum computers, with a steadily increasing number of high-quality qubits, are approaching the threshold of practical utility \cite{ai2024quantum,preskill2025beyond,andersen2025thermalization}. In this pursuit, variational quantum algorithms (VQAs)~\cite{wecker2015progress,cerezo2021variational,bharti2022noisy,peruzzo5variational,tran2024variational,grimsley2023adaptive,jager2023universal,ye2023towards,tian2023recent} have emerged as a leading strategy, attributed to their flexibility in accommodating circuit depth and qubit connectivity among different platforms. In recent years, a wide range of theoretical and experimental studies have demonstrated the feasibility of VQAs across diverse applications, such as quantum chemistry \cite{google2020hartree, kumar2023quantum, guo2024experimental}, optimization \cite{farhi2014quantum,Lukinoptexp,abbas2024challenges}, and machine learning~\cite{Havlcek2019,Liu2021,du2023problem,gujju2024quantum,du2025quantum}. Despite the progress, they face critical challenges when applied to large-scale problems. In particular, the no-cloning theorem and the unitary constraints in the quantum universe prohibit the use of backpropagation techniques common in deep learning \cite{abbas2024quantum}, requiring VQAs to update sequentially to minimize a predefined cost function \cite{schuld2019evaluating,wierichs2022general}. This approach imposes substantial and even prohibitive quantum resource demands, especially in terms of the number of measurements required. Given the scarcity of quantum computers in the foreseeable future, enhancing the optimization efficiency of VQAs while minimizing resource consumption is crucial for enabling their practical deployment.

To advance VQAs for large-scale problems, substantial efforts have been devoted to improving the optimization efficiency. Prior literature in this field can be broadly categorized into three primary classes (see Sec.~\ref{sec:related works} for details). The first aims to reduce measurement costs in quantum many-body and chemistry problems by grouping terms in the Hamiltonian to enable simultaneous measurements \citep{verteletskyi2020measurement, yen2023deterministic,  wu2023overlapped}. The second harnesses classical simulators or learning models to identify well-initialized parameters that are close to local minima of the cost landscape for a given VQA, thereby improving convergence efficiency \cite{galda2021transferability, zhang2025diffusion, PRXQuantum.6.010317}. The third seeks to predict the dynamics of parameter updates by revising past optimization trajectories for the given task~\cite{,du2025artificial}, as exemplified by methods such as recurrent neural networks~\cite{verdon2019learning,kulshrestha2023learning} and QuACK~\cite{luo2024QuACK}. Despite significant advancements, no existing approach effectively balances optimization efficiency and accuracy at scale. In this regard, a critical question arises: \textit{is it possible to achieve both for large-scale VQA systems?}

Towards this question, we observe that prior efforts have primarily focused on hardware improvements and heuristic optimizations, while the potential of approximating training dynamics to alleviate the quantum resource burden \textbf{remains underexplored}. In response, we introduce a fresh perspective by utilizing Taylor expansion to reformulate the parameter optimization process in VQAs as a nonlinear partial differential equation (PDE). In this way,  the evolution of this nonlinear PDE corresponds to the trajectory of parameter updates during training. Building on this formulation, we devise a protocol, dubbed PALQO, that employs a physics-informed neural network (PINN)~\cite{raissi2019physics,karniadakis2021physics} to approximate solutions to the PDE, where high-order terms in the Taylor expansion serve as boundary conditions. The proposed framework is generic, which encompasses the quantum neural tangent kernel (QNTK) as a special case \cite{liu2022representation,tang2022graphqntk,you2023analyzing}. Besides, we prove that a polynomial number of training samples is sufficient to ensure that PALQO attains a satisfactory generalization ability.

We then conduct extensive numerical simulations on different ground state preparation tasks, including the transverse-field Ising model, Heisenberg model, and multiple molecule systems, to investigate the effectiveness of PALQO. Simulation results up to 40 qubits indicate that by learning from a limited set of initial parameter updates obtained from quantum devices, PALQO, which is deployed on classical hardware, can accurately predict future parameter updates. These results indicate the effectiveness of reducing the number of quantum measurements required during optimization. Numerical experiments on ground state preparation tasks across large-scale systems, involving up to 40 qubits, validate the effectiveness of PALQO. Moreover, we show that the proposed PALQO is complementary to existing approaches for improving the optimization efficiency of VQAs. To support the community, we release our source code at \cite{codePALQO}. These results open a new avenue for leveraging the power of PINNs to enhance the efficiency of VQAs and advance the frontier of practical quantum computing.

\textbf{Contributions}. For clarity, we summarize our main contributions below. \textbf{(1)} To our best knowledge, we establish the first general framework between the optimization trajectory of VQAs and PDE, thereby enabling the employment of various PINNs to advance the optimization efficiency of large-scale VQAs. \textbf{(2)} We propose PALQO, an effective PINN oriented to reduce the required number of measurements when training large-scale VQAs, and prove its generalization ability. \textbf{(3)}
Unlike prior studies that mainly focus on small-scale tasks, we conduct extensive numerical studies to validate the advancements of PALQO up to $40$ qubits, providing valuable insights to further improve the optimization efficiency of VQAs at scale.

\section{Preliminaries}\label{sec:preliminary}
In this section, we provide a concise overview of the basic concepts of quantum computing, variational quantum algorithms and physics-informed neural networks to set the stage for their integration in accelerating the quantum optimization process. Please refer to Appendix~\ref{sec:appendix basic concepts} for more details.

\textbf{Basics of Quantum Computing.} 
Quantum state, quantum circuit, and quantum measurement are three key components in quantum computing~\cite{nielsen2010quantum,kaye2006introduction}. In particular, an $n$-qubit quantum state is mathematically represented as a unit vector ${\bf{u}}\in \mathbb{C}^{N}$ in Hilbert space, where $\sum_{j=0}^{N-1}|{\bf{u}}_j|^2=1, N=2^n$. Here we follow conventions to use Dirac notation to represent ${\bf{u}}$ and its transpose conjugate  ${\bf{u}}^\dagger$ , i.e., $|\bf{u}\rangle$ and $\langle{\bf{u}}|$. For a quantum circuit, it serves as a computational model consisting of a sequence of quantum gates that describes operations on the given input state. The most widely used quantum gates are Pauli gates, i.e., $\text{X} = \big(\begin{smallmatrix}
  0 & 1\\
  1 & 0
\end{smallmatrix}\big)$, $\text{Z} = \big(\begin{smallmatrix}
  1 & 0\\
  0 & -1
\end{smallmatrix}\big)$, $\text{Y} = \big(\begin{smallmatrix}
  0 & -i\\
  i & 0
\end{smallmatrix}\big)$. According to the Solovay-Kitaev theorem~\cite{nielsen2010quantum}, arbitrary operation can be approximated by a quantum circuit $U=\prod_j U_j$ where each gate $U_j$ is drawn from a finite universal gate set, such as $\{\text{CNOT}, \text{H}, \text{S}, \text{R}_z(\theta), \text{R}_x(\theta)\}$. Concretely, $\text{CNOT}=|0\rangle\langle0|\otimes\mathbb{I}+|1\rangle\langle1|\otimes \text{X}$, 
$\text{H} = 1/\sqrt{2}\big(\begin{smallmatrix}
  1 & 1\\
  1 & -1
\end{smallmatrix}\big)$, $\text{R}_x(\theta)=e^{-i\theta \text{X}}$, $\text{R}_z(\theta)=e^{-i\theta Z}$, $S = \sqrt{Z}$.

For quantum measurement, it is the process that collapses a quantum state into a definite classical outcome.  In this study, we are interested in the expectation value of the measurement outcomes with a given observable $O$, a Hermitian operator, on quantum state $|{\bf{u}}\rangle$, i.e. $\langle {\bf{u}}| O |{\bf{u}}\rangle$. Suppose the observable presents as an $n$-qubit Hamiltonian that characterizes energy structure of the target quantum system in the form of $H=\sum_{j=1}^{N_H} c_j P_j$, where $P_j$ is a tensor product of Pauli matrices, i.e. $P_j\in\{\mathbb{I}, \text{X}, \text{Y}, \text{Z}\}^{\otimes n}$. To experimentally estimate $\langle {\bf{u}}| O |{\bf{u}}\rangle$ within error $\epsilon$, we typically perform $M\sim \mathcal{O}(1/\epsilon^2)$ repeated measurements for each $P_j$ on multiple copies of the state $|{\bf{u}}\rangle$ and get the outcomes $\{\hat{M}_{\bf{u}}^{j,k}\}_{k=1,..,M}$, then approximate the expectation value by statistical averaging $\langle {\bf{u}}| O |{\bf{u}}\rangle = 1/(MN_H)\sum_{j,k} c_j \hat{M}_{\bf{u}}^{j,k}$.

\textbf{Variational Quantum Algorithms.} Variational quantum algorithms (VQAs)  algorithms designed for machine learning tasks are called quantum neural networks (QNNs)~\cite{massoli2022leap,li2022quantum, hu2022quantum, you2023analyzing, monaco2023quantum, zhao2024quantum,huang2024coreset, nguyen2024theory,Peng2025Tit}, while those applied to many-body physics and quantum chemistry are typically known as variational quantum eigensolvers (VQEs)~\cite{ tilly2022variational,cerezo2022variational,harwood2022improving,ralli2023unitary,sato2023variational, kim2024qudit,wang2023symmetric,chan2024measurement}. The primary objective of VQE is to optimize a parameterized state $|\psi(\bm{\theta})\rangle=U(\bm{\theta})|\phi\rangle$ to minimize the energy function $\mathcal{E}$ defined by a given Hamiltonian $H=\sum_{j=1}^{N_H} c_j P_j$. Mathematically, the energy function to be minimized in VQEs takes the form of  
\begin{equation}
\min_{\bm{\theta}}\mathcal{E}(\bm{\theta}) = \langle \psi(\bm{\theta}) | H | \psi(\bm{\theta}) \rangle. 
\end{equation}
A common and widely adopted approach to complete this optimization problem is utilizing a gradient-based optimizer, like gradient descent, to iteratively adjust the parameters $\bm{\theta}$ according to the partial derivative $\partial_{\bm{\theta}}\mathcal{E}$. Because there is no-clone theorem and no backpropagation without exponential classical overhead in general VQEs \cite{abbas2024quantum}, we need to perform the \textit{parameter shift rule} without involving other quantum resource overhead, such as ancillary qubits to estimate the partial derivative \cite{schuld2019evaluating}. Concretely, the calculation of the partial derivative with respect to $\bm{\theta}_i$ takes the form as
\begin{equation}\label{eq:partial derivative}
\frac{\partial\mathcal{E}}{\partial{\bm{\theta}_i}} = \frac{1}{2}\left[\mathcal{E}\left(\bm{\theta}_i+\frac{\pi}{2}\right) - \mathcal{E}\left(\bm{\theta}_i-\frac{\pi}{2}\right)\right]\approx \frac{1}{MN_H}\sum_{j,k}c_j\left[\hat{M}_{\psi_{+}}^{j,k}+\hat{M}_{\psi_{-}}^{j,k}\right],
\end{equation}
where $\psi_+, \psi_-$ correspond to state $|\psi(\bm{\theta}_i+\pi/2)\rangle$ and $|\psi(\bm{\theta}_i-\pi/2)\rangle$, respectively.

While such a method provides a closed-form expression for gradient estimation without requiring additional qubits and can be extended to general VQEs, it necessitates evaluating  $\mathcal{E}$ twice with shifted parameter values at the same position to estimate the gradient of a single parameter. Hence, suppose the dimension of $\bm{\theta}$ is $p$, it requires to perform $\mathcal{O}(pN_H/\epsilon^2)$ measurements to estimate the partial derivative $\partial_{\bm{\theta}}\mathcal{E}$, which becomes \textbf{computationally prohibitive} in large-scale tasks, especially for large molecules as whose $n$-qubits second-quantized Hamiltonian has roughly $N_H\sim\mathcal{O}(n^4)$ terms \cite{tilly2022variational}. Therefore, it requires substantial resources for estimating updated parameters $\bm{\theta}$ during the optimization, which is considered as one major limitation of large-scale VQEs. Similarly, such a scalability issue also arises in QNNs, where the number of measurements required per iteration scales linearly with $p$ and the batch size.

\textbf{Physical-Informed Neural Network.} 
Physics-informed neural networks (PINNs) have become a promising learning-based tool in approximating the solution of partial differential equations (PDEs) \cite{dissanayake1994neural, lagaris1998artificial, han2017deep}.
With the advantages of computational efficiency for solving complex PDE, they have been widely employed in various practical scenarios such as fluid dynamics, battery degradation modeling, disease detection, and complex systems simulation \cite{karniadakis2021physics,okazaki2022physics,bian2023high,gao2024generative,wang2024physics}. PINNs harness the core tool, automatic differentiation, of modern machine learning to efficiently enforce the physical constraints of the underlying PDE. 

For a PDE problem, it can be generally written as $\mathcal{N}[u(\bm{x},t)]= g(\bm{x},t)$
where \(\bm{x}\in \mathcal{D} \subset \mathbb{R}^d\) denotes variables, \( \mathcal{N} \) represents the differential operators, \( u(\bm{x},t) \) stands for the solution, and $g(\bm{x},t)$ refers to input or source function. The aim of PINNs is to build a neural network $f_{\bm{w}}$ with parameters $\bm{w}$ to approximate the true solution $u$. Hence, the loss function of PINN for solving a general PDE is based on \textit{residuals}, including PDE residual and data residual. The PDE residual measures the difference between the neural network solution and the true solution, expressed as
\begin{equation}\label{eq:origin pde residual}
\mathcal{L}_{P}=\sum_j \left|\mathcal{N}\left[f_{\bm{w}}\left(\bm{x}_p^{(j)},t_p^{(j)}\right)\right] - g\left(\bm{x}_p^{(j)},t_p^{(j)}\right)\right|^2, \quad  \mathcal{L}_{D}=\sum_j \left|f_{\bm{w}}\left(\bm{x}_d^{(j)},t_d^{(j)}\right) - u_d^{(j)}\right|^2.
\end{equation}
Here, $\{\bm{x}_p^{(j)},t_p^{(j)}\}_{j=1}^{N_p}$ used in $\mathcal{L}_P$ are selected collocation points for enforcing PDE structure. For $\mathcal{L}_D$, the dataset 
$\{\bm{x}_d^{(j)},t_d^{(j)}, u_d^{(j)}\}_{j=1}^{N_d}$  with $u_d^{(j)}=u(\bm{x}_d^{(j)},t_d^{(j)})$ denote the training data on $u(\bm{x},t)$ \cite{raissi2019physics}. Thus, the total loss is putting all residuals together, i.e. $\mathcal{L}=\mathcal{L}_{P}+\mathcal{L}_{D}$.
By embedding physical principles into the learning process, PINNs serve as a versatile tool that only requires a small amount of data to tackle the computationally complex problem.

\subsection{Related Works}\label{sec:related works}
Prior literature related to improving the optimization efficiency of VQAs can be classified into three main classes, i.e., \textit{measurement grouping}, \textit{initializer design}, and \textit{prediction of training dynamics}. Since the first two classes are complementary to PALQO, we defer the explanations to Appendix~\ref{sec:appendix related works}. 

The third class aims to harness learning models to approximate the training process. Some works inspired by meta-learning utilize the recurrent neural network to learn a sequential update rule in a heuristic manner \cite{verdon2019learning,kulshrestha2023learning}. Nevertheless, the memory bottleneck and training instability of the recurrent neural network would lead to it being underwhelming \cite{chen2022learning}. Recent work proposed QuACK, involving linear dynamics approximation and nonlinear neural embedding, to accelerate the optimization \cite{luo2024QuACK}. However, the prediction phase requires estimating the energy loss of each step to find the optimal parameters, which is not friendly for large-scale problems. To overcome these limitations, the proposed PALQO uniquely approximates VQA training dynamics using a nonlinear PDE, embedding the dynamical laws directly into the learning process. In this way, it offers \textit{deeper physical insight} and achieves \textit{superior performance} through principled model-guided optimization.

\section{PALQO: physics-informed model for accelerating quantum optimization}

In this section, we first formally define the problem of learning the training dynamics of VQAs as nonlinear PDE problems in Sec.~\ref{sec:reformulate}. Then, in Sec.~\ref{sec:implementation}, we introduce PALQO to solve this PDE via a tailored PINN-based model, where the optimized solutions correspond to the optimization trajectory of VQAs, followed by a generalization error analysis. 

\subsection{Reformulating the optimization of VQA as a PDE problem}\label{sec:reformulate}
Recall the optimization of VQEs in Sec.~\ref{sec:preliminary}.  As it is costly in querying $\bm{\theta}^{(t)}$ of each step $t$, it is demanded to develop a protocol that only learns from a few trajectory data $\{\bm{\theta}^{(t)}\}_{t=1}^{\tau}$ to classically predict future steps, thereby avoiding prohibitive resource costs without compromising accuracy. To achieve this goal, we start by revisiting the gradient descent dynamics of VQEs. The updating rules of parameters $\bm{\theta}$ with learning rate $\eta$ at step $t$ is given by
$\delta{\bm{\theta}}=\bm{\theta}^{(t+1)} - \bm{\theta}^{(t)}  = -\eta \partial\mathcal{E}(\bm{\theta}) / \partial{\bm{\theta}}$, where $\partial\mathcal{E}(\bm{\theta}) / \partial{\bm{\theta}}$ is estimated through phase shift rule shown in Eq.~\eqref{eq:partial derivative}.
Suppose $\eta$ is infinitesimally small, the following ordinary differential equation, a.k.a., gradient flow,  characterizes how parameters change in continuous time, i.e.,  
\begin{equation}\label{eq:dynamic-theta}
\frac{\partial \bm{\theta}}{\partial t}  = -\frac{\partial{\mathcal{E}}}{\partial \bm{\theta}}.
\end{equation}
Besides, we can similarly define the dynamics of $\mathcal{E}$ in a general form under Taylor expansion as
\begin{equation}\label{eq:dynamic-E}
\frac{\partial{\mathcal{E}}}{\partial t}=- \sum_i \left(\frac{\partial\mathcal{E}}{\partial \bm{\theta}_i}\right)^2 + \frac{\eta}{2}\sum_{j,k}\frac{\partial^2\mathcal{E}}{\partial{\bm{\theta}_j}\partial{\bm{\theta}_k}}\frac{\partial\mathcal{E}}{\partial \bm{\theta}_j}\frac{\partial\mathcal{E}}{\partial \bm{\theta}_k}+\mathcal{O}(\eta^2),
\end{equation}
where the first term of RHS of Eq.~\eqref{eq:dynamic-E} is termed as quantum neural tangent Kernel (QNTK) \cite{liu2022representation,tang2022graphqntk,yu2024expressibility}, which captures the sensitivity of outputs to parameter changes, shaping how the gradient flow evolves in parameter space. The second term involves the Hessian matrix of $\mathcal{E}$ that reflects the local curvature of the cost function in the optimization landscape. Thus, tackling the problem of learning optimization trajectory can be recast to solve PDEs presented in Eqs.~\eqref{eq:dynamic-theta} and \eqref{eq:dynamic-E}. This reformulation provides a powerful framework and allows us to leverage a rich set of tools developed for PDEs to understand the underlying behavior of the optimization process in large-scale VQEs. 

\textbf{Remark}. In Appendix~\ref{sec:appendix qntk}, we elucidate the derivation of Eq.~\eqref{eq:dynamic-E} and its relation to QNTK. Besides,  the above reformulation can be effectively extended to broader VQAs such as QNNs.

\subsection{Implementation of PALQO}\label{sec:implementation}

In light of the above reformulation, we propose  PALQO  based on PINN introduced in Eq.~\eqref{eq:origin pde residual} to learn the optimization trajectory of large-scale VQEs. Conceptually, once PALQO attains a low training error, it can predict the optimization path, which substantially reduces the measurement cost as considered in large-scale VQEs. As such, it enhances scalability and resource efficiency by minimizing the need for extensive quantum circuit evaluations while maintaining high fidelity in modeling complex quantum optimization. 

\begin{figure*}[t]
\begin{center}
\centerline{\includegraphics[width=0.92\textwidth]{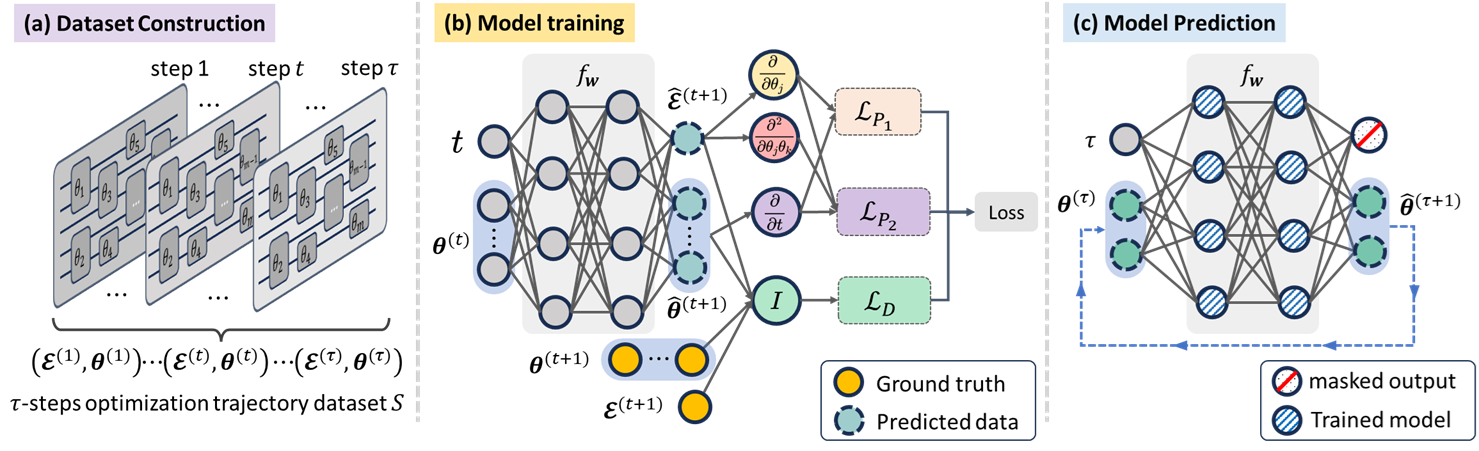}}
\caption{\textbf{Framework of PALQO}. (a) Dataset construction: the tunable angles $\bm{\theta}$ and corresponding loss $\mathcal{E}$ from a VQE over $\tau$ optimization steps are collected, forming a sequential training dataset that captures the optimization trajectory over time. (b) Model training: A PINN $f_{\bm{w}}$ is trained by minimizing the loss $\mathcal{L}=\lambda_{P_1} \mathcal{L}_{P_1}+\lambda_{P_2}\mathcal{L}_{P_2}+\lambda_{D}\mathcal{L}_{D}$. (c) Prediction: starting with the last step $\bm{\theta}^{(\tau)}$, the trained $f_{\bm{w}}$ is used to recurrently predict parameters,  mimicking the optimization process without access to the quantum device.} 
\label{fig:Fig_work_flow}
\end{center}
\end{figure*}

An overview of our protocol is depicted in  Fig.~\ref{fig:Fig_work_flow}, which consists of consists of three stages: \textit{Dataset construction, Model training, and Model prediction}. In Fig.~\ref{fig:Fig_work_flow}~(a), it starts by generating a dataset  $\mathcal{S} = \{\bm{\theta}^{(t)}, \mathcal{E}^{(t)}\}_{t=1}^\tau$ corresponding to trajectory data consisting of $\bm{\theta}$ and energy loss $\mathcal{E}$ over $\tau$ optimization steps collected from a quantum device. Then, as shown in Fig.~\ref{fig:Fig_work_flow}~(b), PALQO utilizes the collected $\mathcal{S}$ to train a PINN-based model $f_{\bm{w}}$ to capture the underlying optimization dynamics. Given the trained $f_{\bm{w}}$, as shown in Fig.~\ref{fig:Fig_work_flow}~(c), it can recursively predict the parameters $\bm{\theta}$ of the future steps until it convergence, i.e. $|\bm{\theta}^{(t)} - \bm{\theta}^{(t+1)}| + |\mathcal{E}^{(t)} - \mathcal{E}^{(t+1)}|\leq \varsigma$ with $\varsigma$ being a small constant.

\textbf{Dataset construction.} To mimic the dynamic of VQE optimization, we need first perform $\tau$ steps optimization via gradient descent to collect a sequential trajectory as the training data, including loss $\{\mathcal{E}^{(1)}, \mathcal{E}^{(2)}, \dots, \mathcal{E}^{(\tau)}\}$ and $\{\bm{\theta}^{(1)}, \bm{\theta}^{(2)}, \dots, \bm{\theta}^{(\tau)}\}$ where $\bm{\theta}^{(j)}=(\bm{\theta}_{1}^{(j)},\cdots,\bm{\theta}_{p}^{(j)})\in \mathbb{R}^p$ is the parameters at $j$-th epoch. Notably, assume a VQE with a total $T$ optimization steps, PALQO only needs $\tau \ll T$ steps to construct the training dataset $\mathcal{S}$. This is because the PINN-based model leverages strong inductive biases from the dynamical laws, such that it does not need to infer fundamental principles from scratch, allowing the model to generalize well from limited data.

\textbf{Model training.} Once the dataset is collected, PALQO employs a deep neural network $f_{\bm{w}}$ that takes  $\{(t,\bm{\theta}^{(t)})\}_{t=1}^{\tau}$ comprising $t$-th time step and trajectory data $\bm{\theta}^{(t)}$ as the input, and predict the loss and parameters for the $(t+1)$-th step, i.e.,$(\mathcal{E}^{(t+1)}, \bm{\theta}^{(t+1)})$.  Refer to Appendix~\ref{sec:implementation_detail} for the details about the neural architecture of  $f_{\bm{w}}$. Denote the prediction of  $f_{\bm{w}}$ for the $t$-th step as $(\hat{\bm{\theta}}^{(t)}, \hat{\mathcal{E}}^{(t)})$. Through leveraging the automatic differentiation capabilities of neural networks, the derivatives of outputs with respect to inputs, i.e., $\partial \hat{\bm{\theta}}/\partial t$ and $\partial \hat{\mathcal{E}}/\partial \bm{\theta}$, can be efficiently computed on classical devices. This efficiency enables the direct incorporation of dynamical law constraints, as described in Eqs.~\eqref{eq:dynamic-theta} and~\eqref{eq:dynamic-E}, into the loss function. 
In this regard, we devise the loss function of PALQO as 
\begin{equation}\label{eqn:overall}
\mathcal{L} = \lambda_D  \mathcal{L}_D + \lambda_{P_1} \mathcal{L}_{P_1} + \lambda_{P_2} \mathcal{L}_{P_2},
\end{equation}
where $\{\lambda_D,\lambda_{P_1},\lambda_{P_2}\}$ denote  hyperparameters to balance the data-driven loss $\mathcal{L}_D$ and two PDE residual losses $\mathcal{L}_{P_1}$ and $\mathcal{L}_{P_2}$. In particular, the data residual loss is defined as $\mathcal{L}_D = \sum_{t=1}^\tau (\mathcal{E}^{(t)} - \hat{\mathcal{E}}^{(t)})^2+\sum_{t=1}^\tau ( \bm{\theta}^{(t)} - \hat{\bm{\theta}}^{(t)})^2$, aiming to capture the temporal changes of $\mathcal{E}$ and $\bm{\theta}$ among $\tau$ steps.
Meanwhile, the two PDE residual losses aim to enforce the PINN to capture the evolution of VQE optimization dynamics through the underlying derivative structure of the loss landscape. Following Eq.~\eqref{eq:dynamic-theta}, the explicit form of the first PDE loss is $\mathcal{L}_{P_1}=\sum_{t=1}^\tau (\sum_{j=1}^p (\partial \hat{\bm{\theta}}^{(t)}_j/ \partial t + \partial \hat{\mathcal{E}}^{(t)} / \partial \bm{\theta}^{(t)}_j  ))^2$. 
By focusing on the first two orders of the derivatives in Eq.~\eqref{eq:dynamic-E}, the second PDE loss yields
\begin{equation}\label{eq:Lp2}
\mathcal{L}_{P_2}=\sum_{t=1}^\tau\Big(\frac{\partial \hat{\mathcal{E}}^{(t)}}{\partial t}+\sum_{j=1}^p \Big(\frac{\partial \hat{\mathcal{E}}^{(t)}}{\partial \bm{\theta}^{(t)}_j}\Big)^2 - \frac{\eta}{2}\sum_{j,k=1}^p \frac{\partial^2 \hat{\mathcal{E}}^{(t)}}{\partial \bm{\theta}^{(t)}_j\partial \bm{\theta}^{(t)}_k} \frac{\partial \hat{\mathcal{E}}^{(t)}}{\partial \bm{\theta}^{(t)}_j}\frac{\partial \hat{\mathcal{E}}^{(t)}}{\partial \bm{\theta}^{(t)}_k}\Big)^2.
\end{equation}
The model $f_{\bm{w}}$ is optimized by minimizing $\mathcal{L}$ in Eq.~\eqref{eqn:overall} via a gradient-based optimizer Adam \cite{kingma2014adam}. 

\textbf{Model prediction.} As the trained  $f_{\bm{w}}$ can not only approximate the solution of the underlying PDE but also capture temporal dependencies of the trajectory data, we are able to recurrently predict the upcoming updates of the $\bm{\theta}$. 
As shown in Fig.~\ref{fig:Fig_work_flow}~(c), by passing $(\tau,\bm{\theta}^{(\tau)})$ through the trained $f_{\bm{w}}$ and masking the $\hat{\mathcal{E}}^{(\tau)}$ node, we can obtain the predicted data $\hat{\bm{\theta}}^{(\tau+1)}$, and then fed $(\tau+1,\hat{\bm{\theta}}^{(\tau+1)})$ as the input back to the network to make the followed prediction. It is worth noting that directly making long-term forecasts to reach the optimal solution of the target problem is exceedingly challenging. To that end, we employ the non-overlapping sliding windows to enhance the network for long-term prediction. More details on the prediction process can be found in Appendix~\ref{sec:implementation_detail}.

\textbf{Remark.} The second PDE residual loss $\mathcal{L}_{P_2}$ can be extended to arbitrary higher orders. Empirical results indicate that a second-order approximation offers a sufficient balance between accuracy and computational cost for modeling the optimization dynamics of the VQEs studied herein. Moreover, while we mainly focus on learning the training process of VQE, our model can be efficiently extended to more general tasks such as quantum machine learning \cite{Havlcek2019,Liu2021,gujju2024quantum}, quantum simulation \cite{yuan2019theory,endo2020variational,kokail2019self}, and quantum optimization \cite{amaro2022filtering,diez2023multiobjective} by slightly modifying the Eqs.~\eqref{eq:dynamic-theta} and~(\ref{eq:dynamic-E}). See Appendix~\ref{sec:appendix additional experiments}\textcolor{red}{} for details.

Building upon prior work on the error analysis of PINNs \cite{de2022error}, we conduct the analysis of the Lipschitz constant bound for PALQO and derive a corollary to establish a generalization bound for PALQO applied to VQEs. An informal statement of the derived generalization bound is provided below, where the formal statement and the related proof are deferred to   Appendix~\ref{sec:appendix theoretical results}.

\begin{corollary}\label{thm:informal main}
(Informal) When utilizing PALQO, whose PINN is constructed by a $L$ layer tanh neural network with most $W$ width of each layer and trained over $\tau$ data samples, to approximate the solution of PDE that describes training dynamics of a VQE with $p$ tunable parameters $\bm{\theta}$, with probability at least $1-\gamma$, its the generalization error is upper bounded by 
\begin{equation}
  \tilde{\mathcal{O}}\left(\sqrt{\frac{pL^2W^2}{\tau}\left(\ln\left(\frac{p}{\epsilon}\right)+\ln\left(\frac{1}{\gamma}\right)\right)}\right).
\end{equation}
\end{corollary}
The achieved results indicate that for any $\epsilon>0$, the number of training examples scales at most polynomially in $p$,  $L$, and   $W$ is sufficient to guarantee a well-bounded generalization error. This warrants the applicability of PALQO in large-scale scenarios.

\section{Experiments}

To evaluate the practical performance of the proposed PALQO framework, we apply it to two representative quantum applications: finding ground state energy of many-body quantum system and molecules in quantum chemistry, which have broad applicability in understanding many-body physical phase transitions \cite{zhang2017observation,liu2024noise,kamakari2025experimental} and simulation of complex electronic structures of molecular systems in drug design and discovery \cite{guo2024experimental,wu2024qvae,vendrell2024designing}. The concrete settings are elucidated below.

\textbf{Many-body quantum system.}
A many-body quantum system consists of interacting quantum particles whose collective behavior and correlations lead to complex phenomena beyond single-particle descriptions. Here, we consider three typical many-body quantum systems.
1)~{\bf{Transverse-field Ising model (TFIM)}} describes spin particles on a lattice interacting via nearest-neighbor coupling and subject to a transverse magnetic field, whose Hamiltonian is typically in a form of $H_{\text{TFIM}} = -J \sum_{j} Z_j \otimes Z_{j+1} - h \sum_j X_j$, where $Z_j$ and $X_j$ refer to the Pauli matrices $Z$ and $X$ applied on the $j$-th qubit, respectively.
2)~{\bf{Quantum Heisenberg model}} also describes the spin particles on a lattice, but spin-spin interactions occur along all spatial directions. Its Hamiltonian can be represented as $H_{\text{QH}}=-1/2\sum_j (J_x X_{j}X_{j+1} + J_y Y_{j}Y_{j+1} + J_z Z_{j}Z_{j+1} + hZ_j)$.
3)~{\bf{Bond-alternating XXZ model}} is an anisotropic variant of the Heisenberg model with unequal coupling strengths in the transverse and longitudinal directions. The Hamiltonian is given by $H_{\text{XXZ}}=\sum_{j=odd}J(X_{j}X_{j+1} + Y_{j}Y_{j+1} + \delta Z_{j}Z_{j+1})+ \sum_{j=even} J'(X_{j}X_{j+1} + Y_{j}Y_{j+1} + \delta Z_{j}Z_{j+1})$. The coefficients $J, h, J_x, J_y, J_z, J', \delta$ within the explored Hamiltonians represent the coupling strength that determines the ground state properties and phase transitions.

\textbf{Quantum chemistry.} The Hamiltonian of a molecule describes the total energy of its electrons and nuclei and serves as the fundamental operator for determining the molecule’s electronic structure in quantum chemistry. The general form of the molecule Hamiltonian can be presented as $H=\sum_j h_j P_j$ where $P_j$ represents tensor products of Pauli matrices, and $h_j$ are the associated real coefficients. Here, we select a widely studied molecule---LiH, and a relatively large and challenging BeH\(_2\) molecule as the target molecules \cite{tilly2022variational,cadi2024folded,tran2024variational}. We generate these molecule Hamiltonians with Openfermion  \cite{Openfermion}. Refer to Appendix~\ref{sec:appendix details of experiments} for more details about the molecule experiments.

\subsection{Experimental Setup}

We employ three standard ansatzes, i.e., hardware-efficient ansatz (HEA) \cite{leone2024practical,wang2024entanglement,zhuang2024hardware}, Hamiltonian variable ansatz (HVA) \cite{wiersema2020exploring,park2024hamiltonian,wang2025imaginary}, unitary coupled cluster with single and double excitations ansatz (UCCSD) \cite{xia2020qubit,boy2025energy,hu2025unitary}, to implement VQE for the different Hamiltonians mentioned above. These ansatzes adopt a layered architecture. Refer to  Appendix~\ref{sec:appendix details of ansatz} for the implementation of these ansatzes. For all ansatzes, their initial parameters $\bm{\theta}^{(0)}$ are uniformly sampled from $[0,1]$, following the strategy adopted in QuACK \cite{luo2024QuACK}. The gradient descent is set as the default optimizer.

For implementing PALQO, we randomly initialize the parameters $\bm{w}$ of the neural network $f_{\bm{w}}$ from $[-1,1]$ and employ the Adam as the optimizer, where the architectures of $f_{\bm{w}}$ are listed in Appendix~\ref{sec:appendix nn architecture}. To improve the training stability and convergence, we utilize a linear decay strategy to adaptively adjust the learning rate during training. Besides, the weight hyperparameters in loss function $\lambda_{P_1}, \lambda_{P_2}, \lambda_{D}, $ are set as $1.0, 1.0$ and $10^{-4}$, respectively. 

\textbf{Benchmark models.} To show the outperformance of PALQO against the state-of-the-art methods, we introduce the following baseline and benchmark. 
First, we use a \textit{vanilla VQE} as the baseline since it provides a well-established reference point to evaluate improvements in accuracy, convergence, and efficiency. 
Second, we select a \textit{LSTM-based model}~\cite{verdon2019learning} as a benchmark since it provides a strong reference for evaluating methods in modeling the temporal dependencies and iterative dynamics of optimization trajectories.
Third, we pick \textit{QuACK}~\cite{luo2024QuACK} as another benchmark as it represents an advanced approach that learns surrogate dynamics of VQAs by embedding Koopman operator-based linear representations into nonlinear neural networks. The implementation of these reference models is deferred to Appendix~\ref{sec: appendix details of benchmarks}.

\textbf{Evaluation metrics.} To quantify the performance of PALQO in accuracy and efficiency, we consider the following metrics,
1)~{\bf{Accuracy.}} we define the accuracy as how close the estimated energy $\hat{E}$ is to a given target energy $E$ of a quantum system, i.e. $\Delta E = |\hat{E} - E|$.
2)~{\bf{Efficiency.}} we define the speedup ratio as $\alpha = {\mathcal{I}_{P}}/{\mathcal{I}_{V}}$, where $\mathcal{I}_{P}$ and $\mathcal{I}_{V}$ refer to the number of iterations required by the baseline method (vanilla VQE) and PALQO or other benchmark models, respectively, to achieve an acceptable accuracy $a$. Specifically, we set $a \le 10^{-3}$.

\subsection{Experimental Results}

We next evaluate the performance of PALQO and other reference models when applied to the aforementioned Hamiltonians under different settings.

\textbf{PALQO significantly reduces the measurement overhead.} 
Here, we utilize the number of measurements incurred during the optimization as a quantum resource measure to explore the performance of PALQO and the other benchmark models when applied to 20 qubits TFIM with HEA, 20 qubits Heisenberg model with HVA, and 14 qubits BeH\(_2\) with UCCSD ansatz, The number of parameters for each case are $120,180,90$, respectively.
As shown in Tab.~\ref{tab:measurement shots}, PALQO achieves significant quantum resource efficiency in aforementioned tasks, 
with around 90\% average reduction in measurement overhead while preserving $\Delta E$ around $10^{-3}$.
\begin{wraptable}{r}{9cm}
\caption{\small The number of quantum measurement shots ($\times 10^8$) required for TFIM, Heisenberg model, and BeH\(_2\). }
\label{tab:measurement shots}
\begin{center}
\begin{small}
\begin{sc}
\begin{tabular}{lccccr}
\toprule
System size & $H_{\text{TFIM}}=20$ &  $H_{\text{HQ}}=20$ & $H_{\text{BeH\(_2\)}}=14$ \\
\midrule
Vanilla VQE & $10.97$ & $21.66$ & $464.3$ \\
LSTM  & $3.126$ & $14.49$ & $312.4$\\
QuACK & $5.217$ & $14.15$ & $461.8$\\
\midrule
PALQO & $\bm{1.535}$ & $\bm{5.749}$ & $\bm{28.01}$\\
\bottomrule
\end{tabular}
\end{sc}
\end{small}
\end{center}
\end{wraptable}
These substantial savings stem from two key factors: 1) compared to the vanilla VQE, PALQO leverages PINN to predict parameter updates, thereby reducing reliance on frequent quantum measurements; 2) the rapid convergence on the classical side enables further reduction in quantum resource expenditure.

\begin{figure}[!ht]
  \centering
  \includegraphics[width=\textwidth]{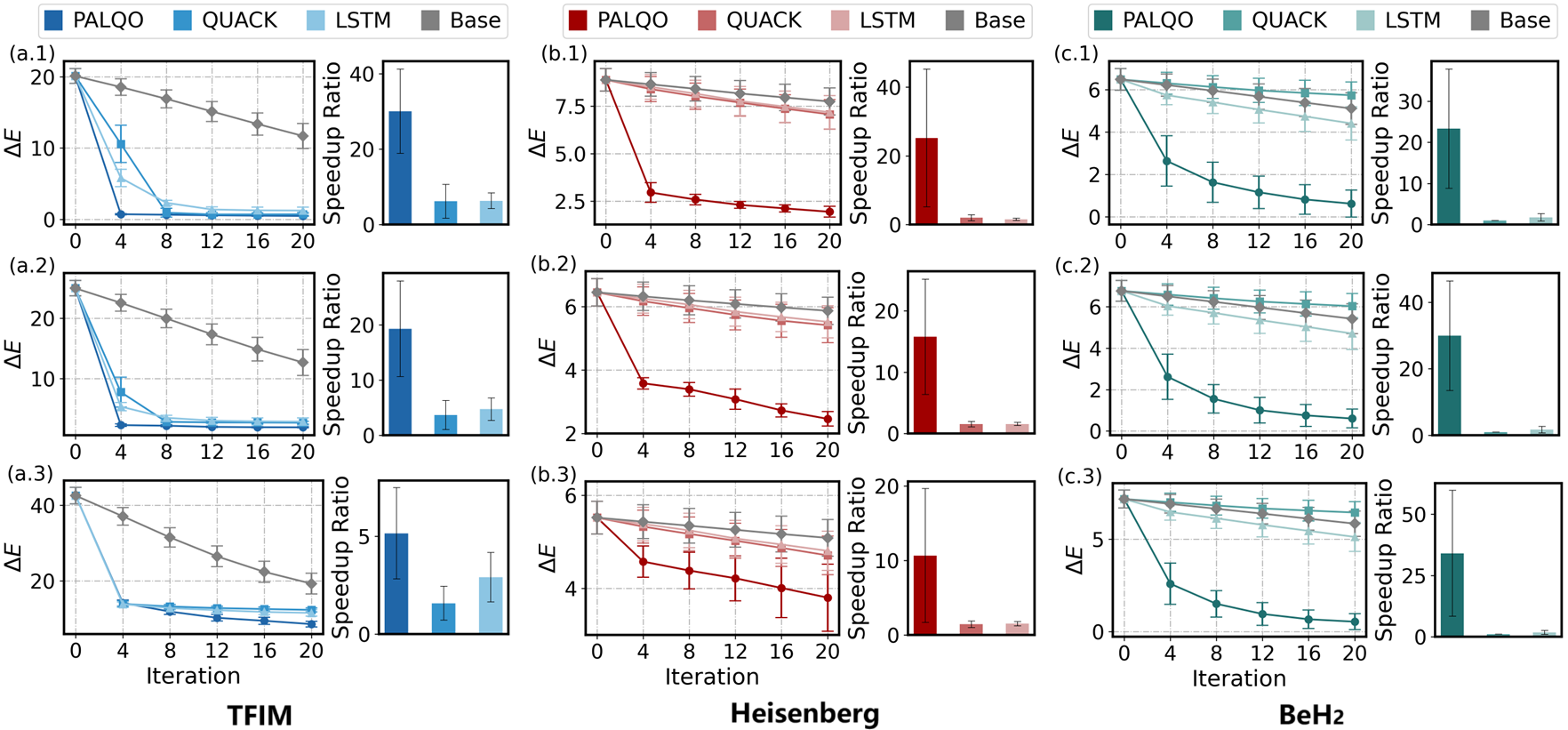}
   \caption{\small Performance comparison between PALQO and the reference models in 20 qubits TFIM with HEA, 20 qubits Heisenberg model with HVA, and 14 qubits BeH\(2\) with UCCSD ansatz. Each subplot comprises a $\Delta E$ curve over iterations performed on a quantum device, along with a bar chart depicting the speedup ratios achieved by PALQO and competing models. The left column illustrates results for TFIM with $J/h=\{2,1,0.5\}$. The central column shows results for Heisenberg model with $J_x=J_y=h=1, J_z=\{0.5, 1, 2\}$. The right column displays the model performance on BeH\(2\) with the bond length $b = \{1.3,1.4,1.5\}$. }
     \label{fig:performance comparison}
\end{figure}

\textbf{PALQO outperforms benchmark models in accuracy and efficiency.} The performance comparisons of PALQO on 20 qubits TFIM with HEA, 20 qubits Heisenberg model with HVA, and 14 qubits BeH\(2\) with UCCSD ansatz for varying structural parameters are presented in Fig.~\ref{fig:performance comparison}. 
In particular, we observed that PALQO consistently outperformed, up to 30x speedup and lower $\Delta E =|\hat{E} - E|$ around $10^{-3}$, like the case of TFIM with $J/h=2$ and HEA ansatz, compared to the other evaluated approaches. 
Furthermore, as the PALOQ  predicts the future optimization steps on classical hardware, it exhibited a faster rate of convergence, achieving a substantial reduction in $\Delta E$ within fewer iterations performed on a quantum device, compared to the baseline methods.
Although the speedup ratio of PALQO has a relatively large variance, its minimum value remains comparable to the average performance of the other approaches. Refer to Appendix~\ref{sec:appendix additional experiments} for the results of XXZ model and LiH. 

\begin{figure*}[!h]
    \centering
    \noindent\makebox[\textwidth][c]{%
        \begin{tabular}{@{}cc@{}}
            \begin{minipage}[t]{0.5\textwidth}
                \noindent\includegraphics[width=7cm, height=3cm]{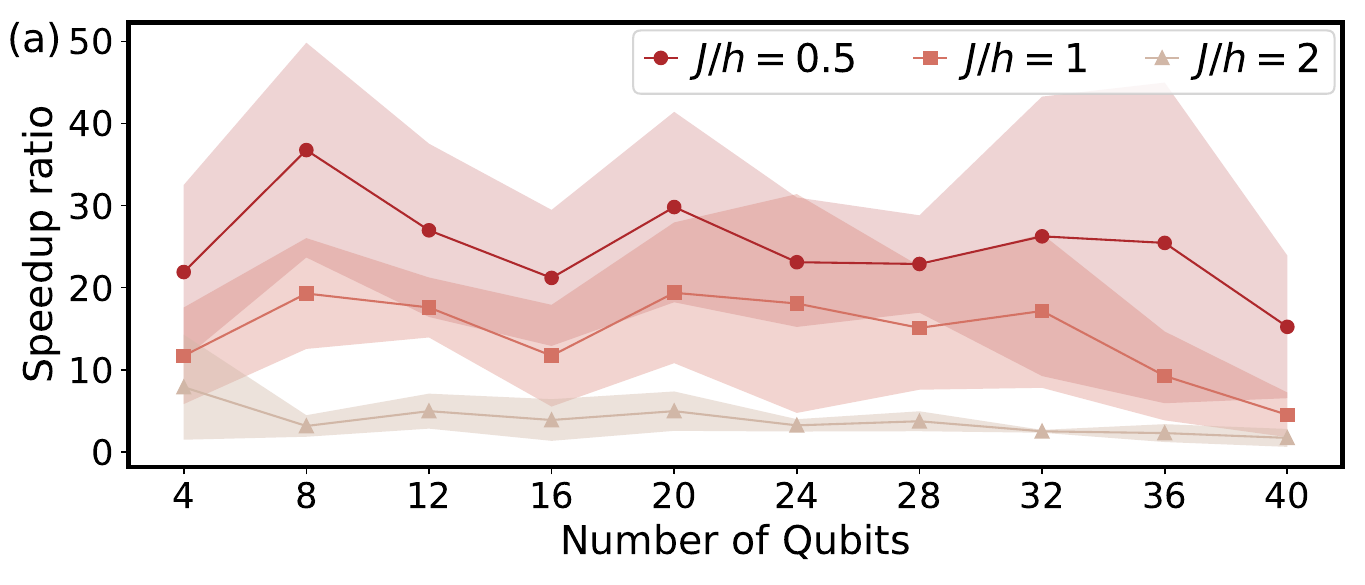}
            \end{minipage} &
            \begin{minipage}[t]{0.5\textwidth}
                \noindent\includegraphics[width=7cm, height=3cm]{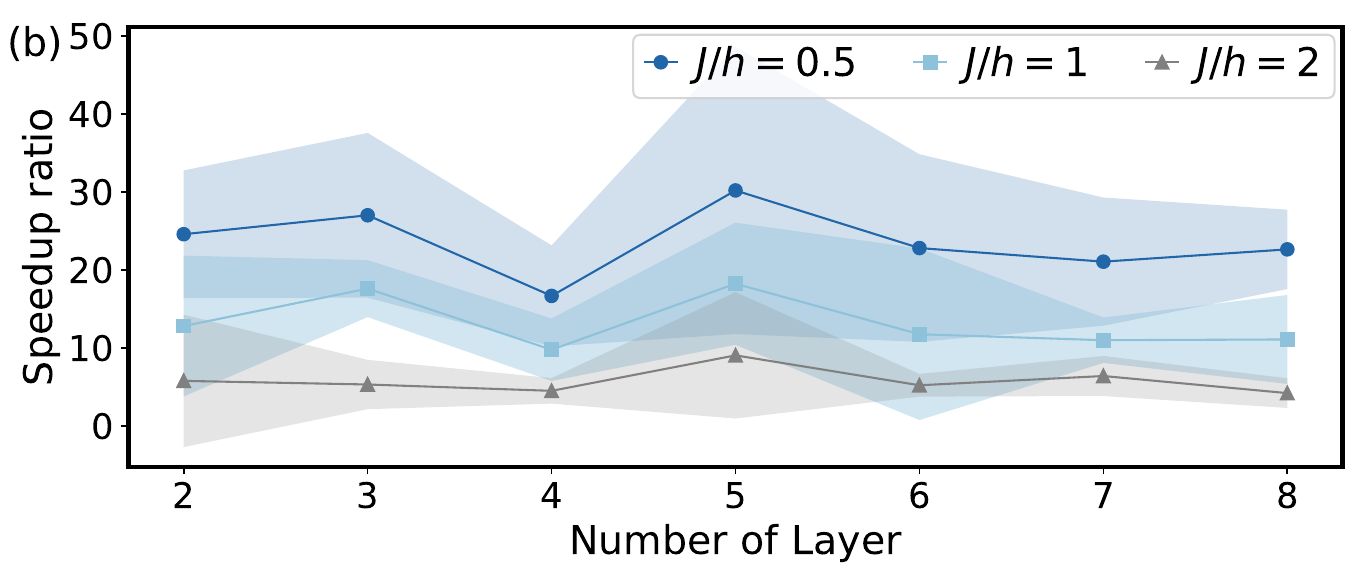}
            \end{minipage} \\
        \end{tabular}%
    }
    \caption{\small Scalability analysis of PALQO on TFIMs. (a) The speedup ratio achieved by PALQO in modeling VQE training dynamics with a fixed HEA, ranging from 4 to 40 qubits. (b) The speedup ratio of PALQO with a fixed system size of 12 qubits, assessed under increasing HEA ansatz layers from 2 to 8.}
    \label{fig:scalability}
\end{figure*}

\textbf{Scalability of PALQO.} We next investigate the scalability of PALQO on the TFIM with HEA, examining its performance in increasing system sizes (from $n=4$ to $n=40$) and the number of ansatz layers (from $2$ to $8$).
In Fig.~\ref{fig:scalability}, the results reveal that the speedup of PALQO is contingent upon the specific system configuration.
Nevertheless, as shown in Fig.~\ref{fig:scalability}~(a), while the speedup ratio fluctuates with the number of qubits varying, it still achieves up to 30x speedup when $J/h=0.5$. The lower speedup at $J/h=2$ is due to a smaller energy gap between the ground and first excited states, making the optimization more challenging.
Similar behavior also appears in Fig.~\ref{fig:scalability}~(b). This suggests that the performance benefits of PALOQ are maintained as the computational demands grow, indicating its potential for large-scale quantum optimization.

\begin{figure}[!h]
    \centering
    \noindent\makebox[\textwidth][c]{%
        \begin{tabular}{ccc}
            \begin{minipage}[t]{0.33\textwidth}
                \noindent\includegraphics[width=\linewidth]{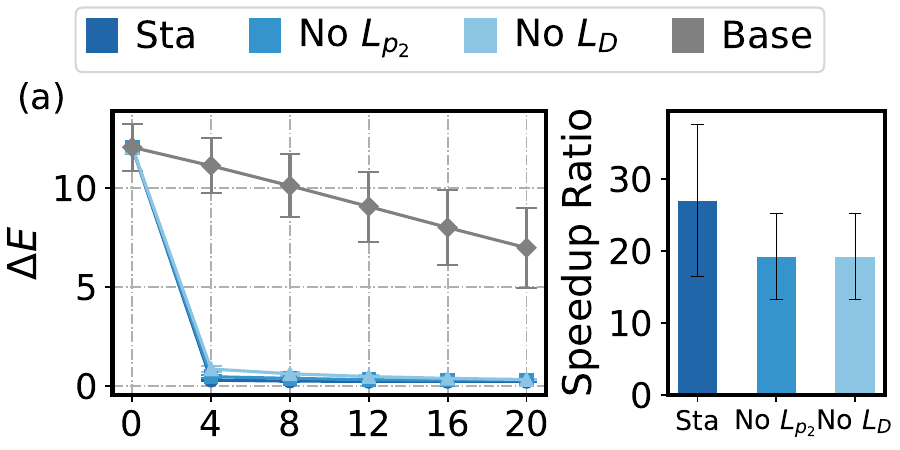}
            \end{minipage} &
            \begin{minipage}[t]{0.33\textwidth}
                \noindent\includegraphics[width=\linewidth]{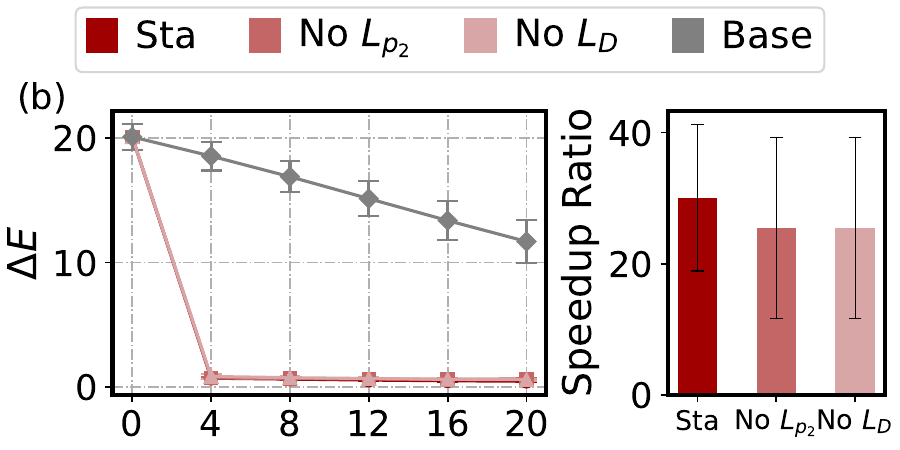}
            \end{minipage}  &
            \begin{minipage}[t]{0.33\textwidth}
                \noindent\includegraphics[width=\linewidth]{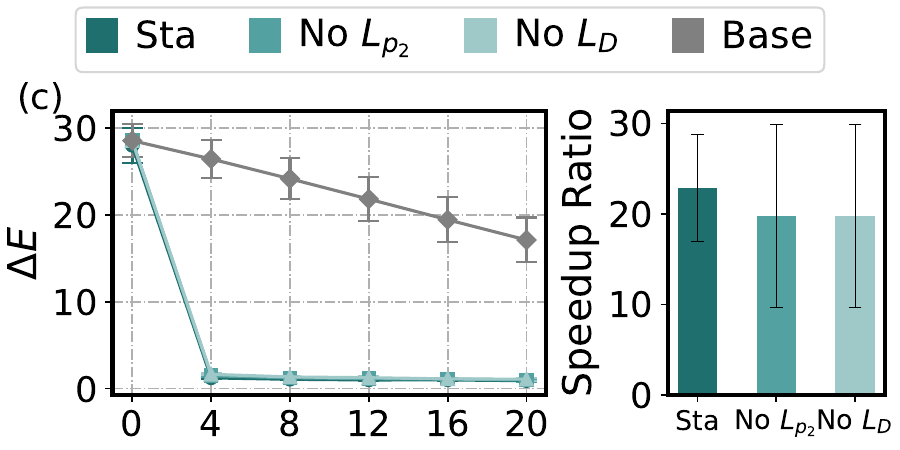}
            \end{minipage} \\
        \end{tabular}%
    }
    \caption{Ablation study on the loss function configuration in PALQO. The three panels show the performance on TFIM with $n=12, 20, 28$ qubits, evaluating the impact of the components $\mathcal{L}_D$ and $\mathcal{L}_{P_2}$ in $\mathcal{L}$.}
    \label{fig:ablation study}
\end{figure}

\textbf{Ablation studies on loss function.} We further evaluate the performance of PALOQ against variants where specific loss terms in Eq.~\eqref{eqn:overall} are removed in the task of TFIM with HEA ansatz. Specifically, we carried out separate ablation studies on the PDE residual and data residual components of the loss function shown in Fig.~\ref{fig:ablation study}. We noticed that while both the PDE and data residual positively influence model performance, their contributions are not essential. These findings suggest that adopting low-order approximations during the construction of PALQO enables the retention of satisfactory speedup while simultaneously reducing the complexity of downstream model training and preventing the degradation of higher-order derivative information.

\begin{figure*}[!h]
    \centering
    \noindent\makebox[0.95\textwidth][c]{%
        \begin{tabular}{@{}cc@{}}
            \begin{minipage}[t]{0.5\textwidth}
                \noindent\includegraphics[width=7cm, height=3cm]{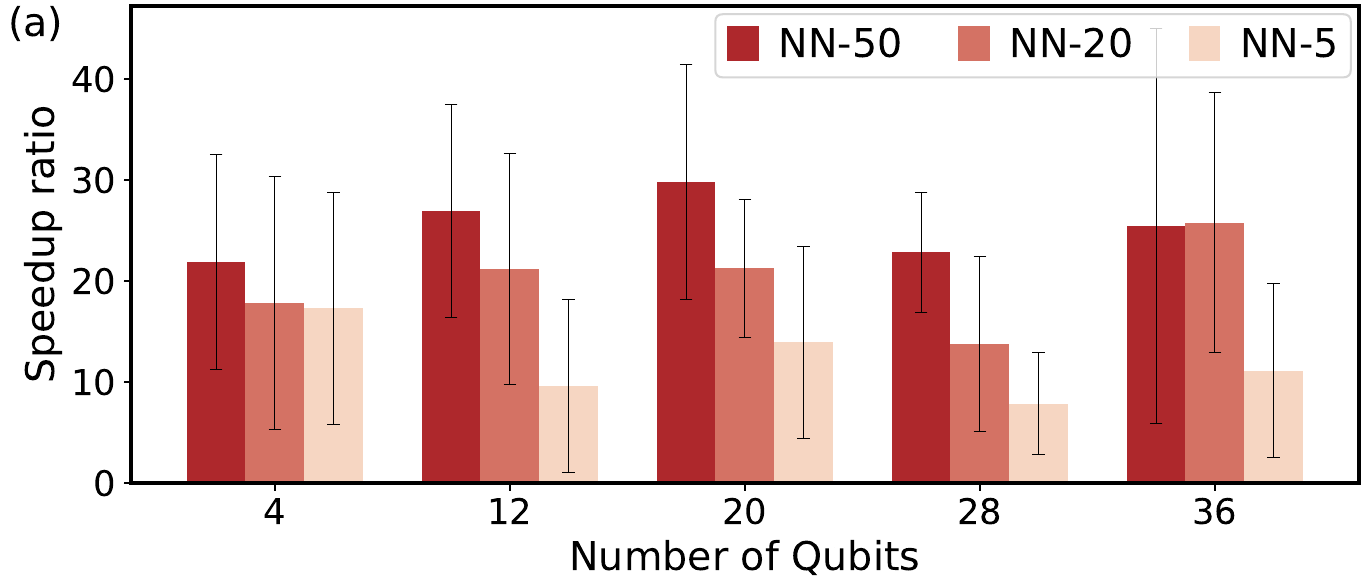}
            \end{minipage} &
            \begin{minipage}[t]{0.5\textwidth}
                \noindent\includegraphics[width=7cm, height=3cm]{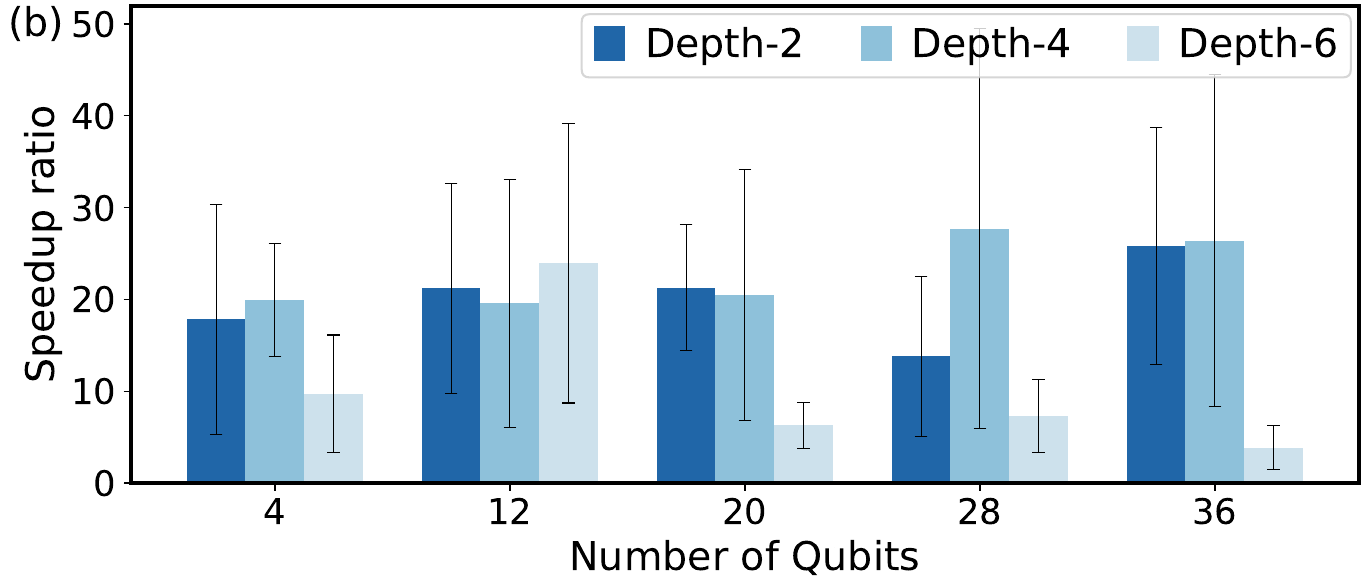}
            \end{minipage} \\
        \end{tabular}%
    }
    \caption{\small Performance of PALQO evaluated on the TFIM under different neural network architecture configurations. (a) Results obtained using 2-layer neural networks with varying hidden layer widths of $W=\{50p, 20p, 5p\}$, where $p$ denotes the number of parameters $\bm{\theta}$. (b) Results with a fixed hidden layer width of $20p$, varying the number of hidden layers from 2 to 6.}
    \label{fig:nn architecture}
\end{figure*}
\textbf{Performance on varying the size of PALQO.}  
To investigate the influence of varying neural network sizes (width $W$ and depth $L$) within PALQO on its performance, we conducted tests on the TFIM with HEA ansatz and $p=6n$ parameters, where $n$ is the system size varying from 4 to 36. In Fig.~\ref{fig:nn architecture}~(a), we observed that an increase in the width of the hidden layers leads to a corresponding improvement in speedup ratio. However,  an inverse phenomenon occurs in Fig.~\ref{fig:nn architecture}~(b), further increasing the neural network depth does not effectively enhance PALQO's performance, which may be related to the vanishing gradient phenomenon, where higher-order derivative information tends to diminish as the neural network becomes deeper \cite{bengio1994learning}. These observations provide guidance for the neural network design in PALQO, indicating that increasing hidden layer width should be prioritized.

\begin{wrapfigure}{r}{8cm}
\centering
\includegraphics[width=6cm, height=2.5cm]{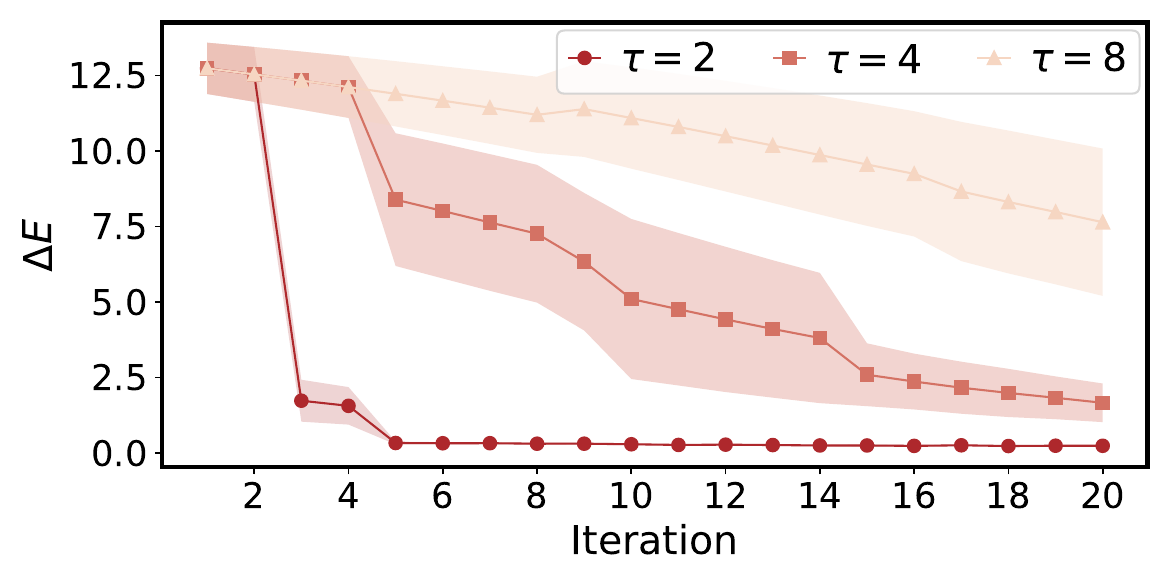}
\caption{\small Performance comparison of PALQO when trained with 2, 4, and 8 data samples.}
\label{fig:sample size}
\end{wrapfigure}
\textbf{The impact of the number of training samples on model prediction.} 
We performed experiments on 12-qubit TFIMs with HEA ansatz using various sample sizes to explore how the number of training samples affects the performance.In Fig.~\ref{fig:sample size}, the results validate that PALQO can achieve satisfactory performance even with a limited number of training samples.
Such data efficiency arises from the direct integration of the physical constraints imposed by the governing PDE into the loss function of the neural network, a core characteristic of PINNs \cite{han2017deep}. 

\section{Conclusion}
In this study, we devised PALQO towards optimizing large-scale VQAs given restrictive quantum resources. In contrast to previous studies, we derive PALQO from reformulating the training dynamics as a nonlinear PDE and using PINN to approximate the solution, and also provide a generalization analysis. Extensive numerical experiments up to 40 qubits validate the effectiveness of PALQO. Although it is still uncertain whether PALQO can scale to the regime where quantum hardware decisively outperforms classical methods, its results at the currently accessible scale are highly encouraging. 
\paragraph{Limitations and future works} 
PALQO reduces the need for repeated quantum gradient evaluations by learning the optimization path classically. While this lowers the number of quantum queries compared to vanilla VQE, it may diminish some quantum advantages. Besides, mitigating the high variance in speedup ratios is crucial for achieving more stable and reliable performance. One future research direction is to incorporate adaptive strategies and variance reduction techniques to achieve this goal and further unlock its potential. 

\section{Acknowledge}
We thank Yusen Wu for the helpful discussions and the anonymous reviewers for their constructive feedback and valuable suggestions. This work was supported by Innovation Program for Quantum Science and Technology (Grant No. 2023ZD0300200), the National Natural Science Foundation of China NSAF (Grant No. U2330201 and No. 92265208). Y. D. acknowledges the support from the SUG grant of NTU. The numerical experiments in this work were supported by the High-Performance Computing Platform of Peking University.

\bibliography{bibliograph.bib}
\bibliographystyle{unsrtnat}

\newpage
\appendix

\section*{Technical Appendices and Supplementary Material}
To facilitate a thorough understanding of our work, this appendix is organized as follows.
First, in Section~\ref{sec:appendix basic concepts}, we introduce the foundational concepts of quantum computing and variational quantum algorithms (VQAs), which form the basis of our work. 
Next, we review the related works focused on improving the optimization efficiency of VQAs in Section~\ref{sec:appendix related works}. 
Then, we detail the implementation of the proposed PALQO, including its connection to the quantum neural tangent kernel (QNTK) and a breakdown of its design components in Section~\ref{sec:implementation_detail}. 
Subsequently, we present the theoretical analysis, covering both generalization error bounds and Lipschitz constant bounds for PALQO in Section~\ref{sec:appendix theoretical results}. 
In addition, we list the experimental details, including computational resources, variational ansätze used in VQE tasks, benchmark descriptions, and experimental setups in Section~\ref{sec:appendix details of experiments}. 
Finally,  in Section~\ref{sec:appendix additional experiments}, we supplement the main results with additional numerical experiments, showcasing PALQO's performance on XXZ and LiH systems, a quantum machine learning task, and the robustness under noise. Besides, we also discussed that it can be complementary to existing approaches, such as measurement grouping, to further improve the optimization efficiency. 
Finally, we discuss the limitations of the proposed method in Section~\ref{sec:appendix additional experiments}.

\section{Quantum Computing and Variational Quantum Algorithms}\label{sec:appendix basic concepts}
\subsection{Basic concepts of quantum computing}
\paragraph{Quantum State}
In quantum computing, the quantum state that stores the information about the physical system is the essential element to be manipulated for computing. We usually describe it as a normalized complex vector in Hilbert space $\mathcal{H}$ by Dirac notation, i.e. $|\psi\rangle \in \mathbb{C}^d$ ($\langle \psi|$ denotes the conjugate transpose of $|\psi\rangle$). For a single-qubit system, as the space $\mathcal{H}=\text{span}(|0\rangle,|1\rangle)$ where $|0\rangle=[1,0]^\top$ and $|1\rangle=[0,1]^\top$, the quantum state $|\psi\rangle$ can be expressed as $|\psi\rangle = \alpha |0\rangle + \beta|1\rangle, \left|\alpha\right|^2+\left|\beta\right|^2 = 1$. Similarly, since the Hilbert space $\mathcal{H}$ of $n$-qubit system  spanned by $\mathcal{H}_1\otimes \dots \otimes \mathcal{H}_n$, an $n$-qubit quantum state $|\psi\rangle$ can be written as $|\psi\rangle=\sum_j \lambda_j |\psi_j\rangle$ where $\sum_{j=1}^2\left|\lambda_j\right|^2=1$, $|\psi_j\rangle = {\otimes_{k=1}^n |\bm{b}_k\rangle}, |\bm{b}_k\rangle \in \{0,1\}^{\otimes N}$.

\paragraph{Quantum Circuit Model}
To process data stored in a quantum state while preserving its normalization under the $l_2$-norm, a unitary transformation $U$ satisfies the requirement that $U^\dagger U = \mathbb{I}$. In quantum computing, the circuit model is a widely used language to describe how the quantum information flows through a network of unitary transformations. 
To process data stored in a quantum state while preserving its normalization under the $l_2$-norm, the unitary transformation $U$ satisfies the requirement such that $U^\dagger U = U U^\dagger = \mathbb{I}$.

\begin{figure}[h]
\includegraphics[width=6cm]{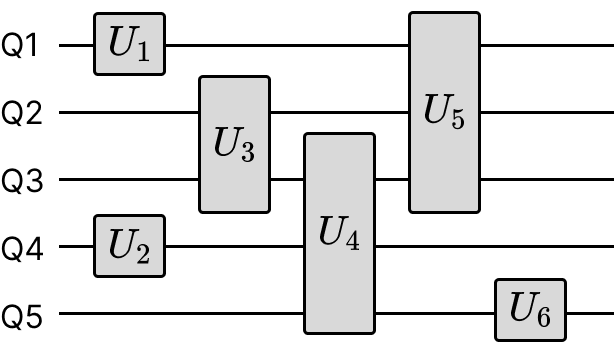}
\centering
\caption{A diagram of a quantum circuit model. The solid block represents the quantum gate, and the horizontal lines stand for qubits. The running order of the quantum circuit is from left to right. The corresponding unitary matrix of this quantum circuit is $U=U_6U_5U_4U_3U_2U_1$.}\label{fig:qc}
\end{figure}

In quantum computing, the circuit model is a widely used language to describe how the quantum information flows through a network of unitary transformations. The diagram of the quantum circuit model is shown in Fig.~\ref{fig:qc}. Like the classical circuit model, we name the unitary operation $U\in \mathbb{C}^{2^n\times 2^n}$ on $n$ qubits as a quantum gate. A group of commonly used single-qubit gates is the Pauli gates, i.e.,
\begin{equation}
I = \begin{bmatrix}
1 & 0 \\
0 & 1 \\
\end{bmatrix},
X =
 \begin{bmatrix}
0 & 1 \\
1 & 0 \\
\end{bmatrix},
Y =
 \begin{bmatrix}
0 & -i \\
i & 0 \\
\end{bmatrix},
Z =
 \begin{bmatrix}
1 & 0 \\
0 & -1 \\
\end{bmatrix}.
\end{equation}
Based on the Pauli gates, there are rotational gates around the $X$, $Y$, $Z$-axes of the Bloch sphere that can be parametrized with the rotation angle $\theta\in\mathbb{R}$, respectively, i.e.,
\begin{equation}
R_x =\begin{bmatrix}
\cos{\frac{\theta}{2}} & -i\sin{\frac{\theta}{2}}  \\
-i\sin{\frac{\theta}{2}}  & \cos{\frac{\theta}{2}}  \\
\end{bmatrix},
R_y =\begin{bmatrix}
\cos{\frac{\theta}{2}} & -\sin{\frac{\theta}{2}}  \\
\sin{\frac{\theta}{2}}  & \cos{\frac{\theta}{2}}  \\
\end{bmatrix},
R_z =\begin{bmatrix}
e^{-i\frac{\theta}{2}} & 0  \\
0  & e^{i\frac{\theta}{2}}  \\
\end{bmatrix}.
\end{equation}
Besides, a widely used multi-qubit gate is a controlled gate which applies a specific operation on the target qubits according to the value of the control qubit, generally formed as $U_c=\left|0\rangle\langle0\right|\otimes \mathbb{I}+\left|1\rangle\langle1\right|\otimes G$ where $G$ is the operation applied on target qubits. The CNOT gate and CZ gate are two specific two-qubit controlled gates where $G$ operation is $X$ or $Z$ gate, respectively. Their mathematical expressions are:
\begin{equation}
\text{CNOT} =  \begin{bmatrix}
1 & 0 & 0 & 0 \\
0 & 1 & 0 & 0 \\
0 & 0 & 0 & 1 \\
0 & 0 & 1 & 0 \\
\end{bmatrix}~\text{and}~
\text{CZ} = \begin{bmatrix}
1 & 0 & 0 & 0 \\
0 & 1 & 0 & 0 \\
0 & 0 & 1 & 0 \\
0 & 0 & 0 & -1 \\
\end{bmatrix}.
\end{equation}

There is a specific collection of quantum gates, termed universal quantum gates, such that any unitary transformation can be represented as a finite sequence of the gates drawn from this set.

\paragraph{Measurement} To extract the classical information from the quantum state, one needs to perform a quantum measurement, which causes the collapse of the superposition into one of its possible states. For instance, when we perform a projective measurement associated with measurement operator $M_m$ where $m$ refers to the measurement outcomes on $|\bf{u}\rangle$, then such an operation returns $m$ with probability $\langle{\bf{u}}|M_m|{\bf{u}}\rangle$.
Besides, through quantum measurement, we can estimate the expectation value of a given Hamiltonian $H$, which corresponds to the average energy of the system in the quantum state $|\psi\rangle$, i.e., $E=\langle \psi|H|\psi\rangle$. 

These components together form the foundation of quantum computation, enabling the execution of quantum algorithms and the realization of quantum advantage. 

\subsection{Variational quantum algorithm}
Variational Quantum Algorithms (VQAs) represent a promising class of hybrid quantum-classical algorithms tailored for the noisy intermediate-scale quantum (NISQ) era \cite{cerezo2021variational,tilly2022variational}. These algorithms cleverly combine the power of quantum computation for preparing and measuring parameterized quantum states with classical optimization routines that iteratively adjust these parameters to minimize a cost function. Generally, the cost function can be expressed as
\begin{equation}\label{eq:cost function vqa}
\mathcal{E}(f,U(\bm{\theta}),\{|\bm{u}\rangle\},\{\bm{O}\},\{{\bm{s}}\})= \sum_{j,k,l} f(\langle \psi(\bm{\theta},\bm{u}_j)|O_k|\psi(\bm{\theta},\bm{u}_j)\rangle,\bm{s}_l),
\end{equation}
where $U(\bm{\theta})$ denotes parametrized quantum circuit with tunable parameters $\bm{\theta}$, $\bm{s}$ refer to labels (optional), $\{|\bm{u}\rangle\}$ and $\{\bm{O}\}$ are a set of given states and observables, respectively, and $|\psi(\bm{\theta},\bm{u}_j)\rangle=U(\bm{\theta})|\bm{u}_j\rangle$ refers to the parametrized quantum state. The following are two typical VQAs: variational quantum eigensolver  (VQE) \cite{peruzzo2014variational, xia2020qubit, tilly2022variational, zhuang2024hardware} and quantum neural network (QNN) \cite{li2022quantum,massoli2022leap,larocca2023theory,you2023analyzing}. 

\paragraph{Variational Quantum Eigensolver} is a prominent variational quantum algorithm specifically designed to find the ground state energy of a quantum system. It utilizes a parameterized quantum circuit to prepare a trial wave function, and a classical optimizer iteratively adjusts the circuit's parameters to minimize the expectation value of the Hamiltonian of the system. Given a Hamiltonian $H=\sum_{k=1}^{N_H} \lambda_k H_k$, the cost function of VQE can be presented in the form of Eq.~\eqref{eq:cost function vqa} by setting $f$ as a identity function, $\{|\bm{u}\rangle\}=\{|0\rangle\}$, $\{\bm{s}\} = \emptyset$, and $\{\bm{O}\}=\{\lambda_k H_k\}_{k=1}^{N_H}$, i.e. 
\begin{equation}\label{eq:cost function vqe}
\mathcal{E}_{\text{VQE}}= \langle 0|U(\bm{\theta})^\dagger H U(\bm{\theta})|0\rangle.
\end{equation}

\paragraph{Quantum Neural Network} is a machine learning model that employs parameterized quantum circuits to learn from data, analogous to the role of layers in classical neural networks \cite{li2022quantum,larocca2023theory}. Given training samples $\{\bm{x}_j,y_j\}^{N}_{j=1}$, the cost function of QNN can be expressed as
\begin{equation}\label{eq:cost function qnn}
\mathcal{
E}_{\text{QNN}}= \frac{1}{2N}\sum_{j=1}^N\Big(\langle \bm{x}_j|U(\bm{\theta})^\dagger O U(\bm{\theta})|\bm{x}_j\rangle - y_j \Big)^2,
\end{equation}
by setting $\{|\bm{u}\rangle\}=\{|\bm{x}_j\rangle\}_{j=1}^N$, $\{\bm{O}\}=\{O\}$, and $\{\bm{s}\} = \{y_j\}_{j=1}^N$, where $f(\cdot, \cdot)$ can be  the mean squared error between $\langle \bm{x}_j|U(\bm{\theta})^\dagger O U(\bm{\theta})|\bm{x}_j\rangle$ and $y_j$.

\section{Related works on accelerating the optimization of VQAs }\label{sec:appendix related works}

\paragraph{Reducing Measurement Costs}.
Since the number of terms in an electronic Hamiltonian generally scales with $\mathcal{O}(N^4)$, where $N$ is the system size, many works explore ways of grouping compatible terms that can be simultaneously measured
\citep{verteletskyi2020measurement,yen2023deterministic,  wu2023overlapped}. 
However, the reduction in measurements heavily relies on the interaction structure of the Hamiltonian, and finding the optimal groups could be computationally complicated.
 
\paragraph{Improving Convergence Efficiency}
Warm start is a common approach that generates superior initializations to improve efficiency in optimization and machine learning. The relevant studies naturally borrow ideas from warm start to enhance the convergence efficiency of VQAs \cite{truger2024warm}. 
One line utilizes problem-specific techniques like randomized rounding in QAOA \cite{egger2021warm}, and imaginary time evolution in QUBO and learning quantum circuit \cite{chai2024structure, PRXQuantum.6.010317}.
In a different vein, some studies focus on exploring generative-based approaches, such as Bayesian Learning \cite{rad2022surviving}, and diffusion model \cite{rad2022surviving}, to identify a promising region in parameter space. Nonetheless, the non-convex landscape of VQA loss appears to be filled with traps \cite{anschuetz2022quantum}.

\paragraph{Predicting Dynamics of Parameter Updates}
Learning to optimize in VQAs aims to harness machine learning to approximate the training process. Some works inspired by meta-learning utilize the recurrent neural network to learn a sequential update rule in a heuristic manner \cite{verdon2019learning,kulshrestha2023learning}. Nevertheless, the memory bottleneck and training instability of the recurrent neural network would lead to it being underwhelming \cite{chen2022learning}. Recent work proposed QuACK, involving linear dynamics approximation and nonlinear neural embedding, to accelerate the optimization \cite{luo2024QuACK}. However, the prediction phase requires estimating the energy loss of each step to find the optimal parameters, which is not friendly for large-scale problems. Our method is developed from an alternative perspective, which explicitly approximates the training dynamics with a second-order nonlinear PDE, then utilizes a learning-based model to find the solution.


\section{Implementation Details of PALQO}\label{sec:implementation_detail}
In this section, we present a more detailed discussion about the PALQO, including the relation to QNTK, and details of the training and prediction process. 
 
\subsection{Relation to QNTK}\label{sec:appendix qntk}
The quantum neural tangent kernel (QNTK) is a tool used to analyze the behavior of VQAs, particularly variational quantum circuits \cite{liu2022representation, you2023analyzing}. Inspired by the neural tangent kernel from classical deep learning, the QNTK allows for theoretical insights into the training dynamics and generalization properties of these quantum models.

Let us first present the explicit form of QNTK in QNN. Recall the definition of QNN in Eq.~\eqref{eq:cost function qnn}, where the loss function is $\mathcal{E}_{\text{QNN}}$, the number of trainable parameters is $p$. Let the residual of $j$-th sample be $\mathcal{E}_j=g(\bm{x}_j, \bm{\theta}) - y_j$ where $g(\bm{x}_j,\bm{\theta})  = \langle \bm{x}_j|U(\bm{\theta})^\dagger O U(\bm{\theta})|\bm{x}_j\rangle$. The derivative of $\mathcal{E}_j$ with respect to $t$ can be expressed as
\begin{equation}\label{eq:dynamics qnn}
\frac{\partial \mathcal{E}_i}{\partial t} = -\frac{\eta}{2N}\sum_{j=1}^N\sum_{k=1}^p  \frac{\partial g(\bm{x}_i,\bm{\theta}^{(t)})}{\partial\bm{\theta}_k^{(t)}}\frac{\partial g(\bm{x}_j,\bm{\theta}^{(t)})}{\partial \bm{\theta}_k^{(t)}}\mathcal{E}_j.
\end{equation} 
In this regard, the element of QNTK, $K_{i,j}$, is defined as
\begin{equation}\label{eq:qntk}
K_{i,j} \equiv \sum_{k=1}^p \frac{\partial g(\bm{x}_i,\bm{\theta}^{(t)})}{\partial\bm{\theta}_k^{(t)}}\frac{\partial g(\bm{x}_j,\bm{\theta}^{(t)})}{\partial \bm{\theta}_k^{(t)}}.
\end{equation}

We next present QNTK in VQE. We consider the cost function of VQE in Eq.~\eqref{eq:cost function vqe}, the change between every two iterations can be expressed as
\begin{align}
\Delta \mathcal{E}_{\text{VQE}} &= \mathcal{E}_{\text{VQE}}(\bm{\theta}^{(t+1)}) - \mathcal{E}_{\text{VQE}}(\bm{\theta}^{(t)}),\\
&= \mathcal{E}_{\text{VQE}}(\bm{\theta}^{(t)}+\delta \bm{\theta}^{(t)}) - \mathcal{E}_{\text{VQE}}(\bm{\theta}^{(t)}).
\end{align}
Supported by Taylor expansion, we have
\begin{equation}
\mathcal{E}_{\text{VQE}}(\bm{\theta}^{(t)}+\delta \bm{\theta}^{(t)}) =  \mathcal{E}_{\text{VQE}}(\bm{\theta}^{(t)})  + \sum_i \frac{\mathcal{E}_{\text{VQE}}}{\partial \bm{\theta}^{(t)}_i}\delta \bm{\theta}^{(t)}_i
 + \frac{1}{2}\sum_{j,k} \frac{\partial^2 \mathcal{E}_{\text{VQE}}}{\partial \bm{\theta}^{(t)}_j \partial \bm{\theta}^{(t)}_k}  \partial \mathcal{E}\delta \bm{\theta}^{(t)}_j \delta \bm{\theta}^{(t)}_k +\mathcal{O}(\|\delta \bm{\theta}^{(t)}\|^3).
\end{equation}
Since $\delta \bm{\theta}^{(t)}=-\eta \nabla_{\bm{\theta}}\mathcal{E}_{\text{VQE}}(\bm{\theta}^{(t)})$, suppose the learning rate $\eta$ is infinitesimally small, we can write the dynamics of $\mathcal{E}_{\text{VQE}}$ as
\begin{equation}\label{eq:dynamics vqe}
\frac{\partial \mathcal{E}_{\text{VQE}}}{\partial t} = - \sum_i \frac{\partial \mathcal{E}_{\text{VQE}}}{\partial \bm{\theta}^{(t)}_i}\frac{\partial \mathcal{E}_{\text{VQE}}}{\partial \bm{\theta}^{(t)}_i} + \frac{1}{2}\eta\sum_{j,k} \frac{\partial^2 \mathcal{E}_{\text{VQE}}}{\partial \bm{\theta}^{(t)}_j \partial \bm{\theta}^{(t)}_k}\frac{\partial \mathcal{E}_{\text{VQE}}}{\partial \bm{\theta}^{(t)}_j}\frac{\partial \mathcal{E}_{\text{VQE}}}{\partial \bm{\theta}^{(t)}_k}+\mathcal{O}(\eta^2).
\end{equation}
In Eq.~\eqref{eq:dynamics vqe}, the first contributing term can be regarded as a special case of QNTK in Eq.~\eqref{eq:qntk} that only has a single data point, denoted as
\begin{equation}
K' = \sum_i \frac{\partial \mathcal{E}_{\text{VQE}}}{\partial \bm{\theta}^{(t)}_i}\frac{\partial \mathcal{E}_{\text{VQE}}}{\partial \bm{\theta}^{(t)}_i}.
\end{equation}
This suggests that due to the similarity in cost functions of various VQAs, PALQO can be naturally extended to other VQA models like QNNs.

\subsection{Implementation details of PALQO}

PALQO is a hybrid quantum-classical algorithm designed to optimize VQA parameters by iteratively combining short VQA training runs on a quantum device with classical learning using a PINN. In each iteration, the algorithm performs a few VQA steps to gather data (i.e., $\mathcal{\bm{\theta}}$ and $\mathcal{E}$), trains the PINN to model the local loss landscape, and then uses the trained PINN to predict a potentially better set of parameters. These predicted parameters are then used as the starting point for the next VQA training phase, repeating the cycle until the VQA loss converges, aiming to accelerate and improve the overall optimization process by leveraging the PINN as a surrogate model to guide the search in the parameter space. The whole process of PALQO is summarized in Algorithm~\ref{alg:PALQO}.

\begin{algorithm}[ht]
   \caption{PALQO}
   \label{alg:PALQO}
\begin{algorithmic}[1]
    \STATE {\bfseries Input:} a VQA with parameters $\bm{\theta}$, PINN-based model $f_{\bm{w}}$ with $\bm{w}$ constituting weights and biases.
    \STATE {\bfseries Output:} Parameters \(\hat{\bm{\theta}}^*\) to minimize the VQA loss.
    \STATE Randomly initialize the $\bm{\theta}$ and $\bm{w}$.
    \REPEAT
    \STATE Perform $\tau$ steps VQA training on quantum device to form $\mathcal{S} = \{\bm{\theta}^{(t)}, \mathcal{E}^{(t)}\}_{t=1}^\tau$.
    \STATE Train the model $f_{\bm{w}}$ over $\mathcal{S}$.
    \STATE $j \leftarrow 0$, $\bm{\theta}^{(j)} \leftarrow \bm{\theta}^{(\tau)}$
    \REPEAT 
    \STATE $\hat{\bm{\theta}}^{(j+1)} = f_{\bm{w}}(\bm{\theta}^{(j)})$, $\bm{\theta}^{(j)} \leftarrow \hat{\bm{\theta}}^{(j+1)}$.
    \UNTIL {$\hat{\bm{\theta}}^{(j)}$ converge} 
    \STATE $\hat{\bm{\theta}}^* \leftarrow \hat{\bm{\theta}}^{(j)}, \bm{\theta}\leftarrow \hat{\bm{\theta}}^*$.
    \UNTIL {$\mathcal{E}_{\bm{\theta}}$ and $\bm{\theta}$ converge}
    \STATE {\bfseries Return: $\hat{\bm{\theta}}^*$} 
\end{algorithmic}

\end{algorithm}

Instead of relying solely on the potentially noisy and gradient-limited information obtained directly from the quantum device in each step, PALQO uses the PINN to learn a smoother and more global picture of the loss landscape based on local explorations. This can potentially lead to faster convergence and help escape local minima in VQA optimization.

In the following, we elucidate the implementation of each step omitted in the main text.

\paragraph{Dataset Collection}  Here, we formally define the dataset required for each training session. The dataset consists of $m$ sets, each corresponding to one step in the VQE iteration. Each training sample consists of an input-output pair $\{(t, \bm{\theta}^{(t)}), (\mathcal{E}^{(t)}, \bm{\theta}^{(t+1)})\}$, where $\bm{\theta}^{(t)}$ represents the variational parameters at step $t$, and $\mathcal{E}^{(t)}$ is the corresponding loss function value. The variable $t$ is a custom-defined discrete sequence that maintains the temporal ordering of $\bm{\theta}^{(t)}$. To ensure consistency, here we specify the input-output par as $\{(\hat{t}, \bm{\theta}^{(t)}), (\mathcal{E}^{(t)}, \bm{\theta}^{(t+1)})\}$ where $\hat{t}$ is the time variable starting at $0.01$ and increases by $0.01$ for each step $t$. In other words, for a dataset with $\tau$ training samples, $\hat{t}$ takes values from $0.01$ to $0.01 \times \tau$. 

\paragraph{Neural Network Structure}
\label{sec:appendix nn architecture}
The Neural Network is a fully connected feedforward neural network with two hidden layers. The total number of variational parameters is defined as $p$, making both the input and output dimensions $p + 1$. Each hidden layer consists of $50 \times p$ neurons, and the activation function for all layers is \texttt{tanh}.

\paragraph{Iterative Prediction in PALQO}
As described in the main text, the prediction process involves feeding the input $(t+\tau, \bm{\theta}^{(t+\tau)})$ into the network to iteratively produce the $m$-step prediction, i.e. $\{\hat{\bm{\theta}}^{(t+\tau+j)}\}_{j=1}^m$. And the iterative prediction terminates once $\hat{\bm{\theta}}$ converges. Here, calculating the $\mathcal{E}$ in each step is expensive, thereby the convergence is defined as satisfying the condition only on $\bm{\theta}$: $\Delta = \|\hat{\bm{\theta}}^{(t+\tau+m)} - \hat{\bm{\theta}}^{(t+\tau+m-1)}\|_2 < \epsilon$, where $\epsilon = 10^{-4}$.
However, in the actual VQE optimization trajectory, $\Delta$ tends to decrease gradually as iterations progress. If the stopping condition is applied directly, it may lead to premature termination, resulting in suboptimal performance, or excessively delayed termination, leading to unnecessary computational overhead. 

To address this issue, we incorporate an additional guarantee mechanism: the iterative prediction is executed for a fixed number of 2000 iterations. We separately calculate the loss $\mathcal{E}$ using the $\bm{\theta}$ that minimizes $\Delta$ and $\hat{\bm{\theta}}^{(t+\tau+2000)}$, and then select the minimal one as the optimal variational parameter, $\hat{\bm{\theta}}^*$, which is subsequently used as the initialization $\bm{\theta}^{(0)}$ for the next VQE cycle.

\section{Theoretical Analysis}\label{sec:appendix theoretical results}
In this section, we provide a rigorous analysis of the performance of PALQO, which builds on a previous work, i.e,. Corollary 1 of (De Ryck \& Mishra (2022))~\cite{de2022error}, to gain insights into the generalization ability of PALQO. Notably, while previous work offers a general bound, it cannot be directly applied to the nonlinear PDEs relevant to our problem. Therefore, we introduce Lemmas \ref{lemma:bound_J_tanhNN}, \ref{lemma:bound_H_tanhNN},  and \ref{lemma:lipsitch_E}, and combine these with Corollary 1 in Ref.~\cite{de2022error} to derive our Corollary \ref{thm:ge_pinn}.

\begin{lemma}\label{lemma:bound_J_tanhNN}
Given an $L$-layer tanh neural network $f(\bm{x},(\bm{W},\bm{b}))$ constructed by bounded weights $\bm{W}=\{W^{(l)},|W^{(l)}|\leq a, l\in[L]\}$, bias $\bm{b}=\{b^{(l)},|b^{(l)}|\leq a, l\in[L]\}$ and activation function $\sigma=tanh(x)$, the norm of Jacobian with respect to input vector $\bm{x}$ is bounded by,
\begin{align*}
\left|J_f\right| \leq a^L.
\end{align*}
\end{lemma}
\begin{proof}
As the output of $l$-layer can be presented by $\bm{f}_l = \sigma\left({W^{(l)}}^\top \bm{f}_{l-1}+b^{(l)}\right)$ and $\sigma'(x)=1-\sigma^2(x)$, the Jacobian with respect to the input vector is
\begin{equation}
J^{(l)}=\frac{\partial{\bm{f}_{l}}}{\partial{\bm{f}_{l-1}}}=diag[\sigma'(\bm{f}_{(l-1)})]\cdot{W^{(l)}}^\top.
\end{equation}
According to the chain rule, we can derive the Jacobian of $f(\bm{x},(\bm{W},\bm{b}))$ as 
\begin{equation}
J_f = \prod_{l=0}^{L-1} J^{(L-l)}=\prod_{l=0}^{L-1} diag[\sigma'(\bm{f}_{(L-l-1)})]\cdot{W^{(L-l)}}^\top.
\end{equation}
Since $\sigma'=sech^2(\bm{x})$ and let $D=diag(\sigma')$, we have $|D_{i,i}| \leq 1$. Then, as $\left|W^{(l)}\right|\leq a$, we have
\begin{equation}
\left|J_f\right| \leq a^L.
\end{equation}
\end{proof}

\begin{lemma}\label{lemma:bound_H_tanhNN}
For an $L$-layer tanh neural network $f(\bm{x},(\bm{W},\bm{b}))$ constructed by  bounded weights $\bm{W}=\{W^{(l)},|W^{(l)}|\leq a, l\in[L]\}$, bias $\bm{b}=\{b^{(l)},|b^{(l)}|\leq a, l\in[L]\}$ and activation function $\sigma=tanh(x)$, the norm of Hessian with respect to input vector $\bm{x}$ is bounded by,
\begin{equation}
\left|H_f\right|\leq 2a^{2L}L.
\end{equation}
\end{lemma}

\begin{proof}
Since $\sigma'(x)=1-\sigma^2(x)$ and $\sigma''(x)=-2\sigma(x)(1-\sigma^2(x))$, the Hessian of $f(\bm{x},(\bm{W},\bm{b}))$ can be expressed as
\begin{equation}
H^{(l)}= \frac{\partial^2\bm{f}_l}{\partial{(\bm{f}_{l-1})^2}}= { diag[\sigma''(\bm{f}_{l-1})] \cdot  W^{(l)} W^{(l)}}^\top.
\end{equation}
According to the lemma of expression for Hessian $H$ in terms of $J$ \citep{de2022error}, and $\left|\sigma''(x)\right|\leq 2$
\begin{equation}
H_f = \sum_{l=1}^L {J^{(1)}}^\top\cdots {J^{(l-1)}}^\top \cdot\left(J^{(L)} \cdots J^{(l+1)} H^{(l)}\right)\cdot J^{(l-1)}\cdots J^{(1)}.
\end{equation}
we can bound the $H_f$ by
\begin{equation}
\left|H_f\right|\leq 2a^{2L}L.
\end{equation}
\end{proof}

\begin{lemma}[Lipschitz
continuous of Jacobian and Hessian (Lemma 12, \citep{de2022error})]\label{lemma:lip_J_H_tanhNN}
Let $a,b \in \mathbb{R}$, for an $L$-layer tanh neural network $f(\bm{x},(\bm{W},\bm{b}))$ constructed by  bounded weights $\phi \in \{\bm{W},\bm{b}\}, \left|\phi\right|\leq a$ and activation function $\sigma=tanh(x)$, at most $W$ width, it holds that for any $\bm{x}\in[-b,b]^p$,
\begin{align*}
\left|J_\phi - J_{\phi'}\right| & \leq b(p+7)L a^{2L-1}W^{2L-2}2^{L}\left|\phi - \phi'\right|, \cr
\left|H_\phi - H_{\phi'}\right| & \leq b(p+7)L^2 a^{3L-1}W^{3L-3}2^{L+2}\left|\phi - \phi'\right|.
\end{align*}
\end{lemma}

\begin{lemma}\label{lemma:lipsitch_E}
Let $a,b,N \in \mathbb{R}$, suppose that the employed PINN is constructed by the tanh neural network with bounded weights and biases $\phi \in [-a,a]^m$, at most $L$ layers and $W$ width. Moreover, suppose it adopts a smooth activation function $\sigma=tanh(x)=\frac{e^{-x}-e^x}{e^{-x}+e^x}$, and the input $\bm{x} = \{x_j\}_{j=1}^N$ where $\bm{x}_j \in[-b,b]^p$. When applying such a PINN to approximate the solution of training dynamics of VQAs with a fixed learning rate $\eta$. The Lipschitz constant $\mathcal{L}$ of training error $\mathcal{E}_T$ or generalization error $\mathcal{E}_G$ can be respectively bounded by
\begin{align}
\mathcal{L} & \leq \mathcal{O}\left(\text{poly}(b,p,L,\eta,a^L,W^L)\right).
\end{align}
\end{lemma}
\begin{proof}
Since the analysis of $\mathcal{L}$ of $\mathcal{E}_T$ and $\mathcal{E}_T$ is similar, here we mainly focus on $\mathcal{E}_T$. As we select the square error as the loss function, i.e. 
\begin{equation}\label{eqn:lemmaD4-1}
\mathcal{E}_T(\phi) = \frac{1}{N}\sum_{j=1}^N \left(\mathcal{R}[f_\phi(x_j)]\right)^2=\frac{1}{N}\sum_{j=1}^N(\partial_t f_{\phi}(x_j) - \mathcal{N}[f_{\phi}(x_j)])^2,
\end{equation}
where $\mathcal{R}$ is residual of PDE, and $f_\phi$ is the PINN approximation. As $\mathcal{E}_T$ is differentiable, we have
\begin{equation}
\left|\mathcal{E}_{T}(\phi) - \mathcal{E}_T(\phi')\right| \leq 2\max_{\phi}\left|\mathcal{R}[f_{\phi}]\right|\left|\mathcal{R}[f_{\phi}]-\mathcal{R}[f_{\phi'}]\right|.
\end{equation}
For the $\left|\mathcal{R}[f_{\phi}]-\mathcal{R}[f_{\phi'}]\right|$ term, according to the chain rule for the derivative of a composite function, we have $J_\phi=\prod_{k=0}^{L-1}J_\phi^{L-k}$, $H_\phi=\sum_{k=0}^L (J_\phi^1)^\top\cdots(J_\phi^{k-1})^\top\cdot(J_\phi^L\cdots J_\phi^{k+1}H_{\phi}^k)\cdot J_\phi^{k-1}\cdots J_\phi^1$, where $J_\phi^{L-k}$ is the Jacobinan matrix at the $(L-k)$- the layer, and $H_\phi^k$ is the Hessian matrix at the $k$-th layer. For the training dynamic of VQAs with a fixed learning rate $\eta$, we can formulate it as a PDE as shown in Eq.~\eqref{eq:dynamic-E}, i.e. 
\begin{align}
\mathcal{N}[f_\phi] 
&= J_\phi^\top \cdot J_\phi - \frac{1}{2}\eta J_{\phi}^\top\cdot H_{\phi} \cdot J_{\phi}.
\end{align}
where $\mathcal{N}$ is the differential operator. As $\partial_t f_{\phi}$ can also be regarded as the Jacobian only for the variable $t$.
Thus, we have
\begin{align}
\left|\mathcal{R}[f_{\phi}]-\mathcal{R}[f_{\phi'}]\right| \leq  \left|J_{\phi} - J_{\phi'}\right| + \underbrace{\left|J_\phi^\top\cdot J_{\phi} - J_{\phi'}^\top\cdot J_{\phi'}\right|}_{A}+\frac{1}{2}\eta \underbrace{\left|J_{\phi'}^\top\cdot H_{\phi'}\cdot J_{\phi'} - J_{\phi}^\top\cdot H_{\phi} \cdot J_{\phi} \right|}_{B}.
\end{align}
Since the activiation function $\sigma=tanh(x)$, $|\sigma'|_{\infty} = 1 $ and $|\sigma''|_{\infty}\leq 1$, and based on the Lemma of Lipschitz continuity of Jacobian and Hessian (Lemma \ref{lemma:lip_J_H_tanhNN}), we can bound the $A$ term
\begin{align}
A&=\left|J_\phi^\top\cdot J_{\phi} - J_{\phi'}^\top\cdot J_{\phi'}\right| \leq \left(\left|J_{\phi}^\top\right| + \left|J_{\phi'}^\top \right|\right) \left|J_{\phi} - J_{\phi'}\right| \cr
&\leq b(p+7)L a^{3L-1}W^{2L-2}2^{L+2}\left|\phi - \phi'\right|.
\end{align}
Similarly, the term $B$ can be bounded by
\begin{align*}
B&=\left|J_{\phi'}^\top\cdot H_{\phi'}\cdot J_{\phi'} - J_{\phi}^\top\cdot H_{\phi} \cdot J_{\phi} \right|\cr
&\leq |J_{\phi'}^\top-J_{\phi}^\top|\left|H_{\phi'} - H_\phi\right|\left|J_{\phi'}-J_{\phi}\right| + \left|J_{\phi'}^\top\right| \left|H_{\phi'}-H_{\phi}\right| \left|J_{\phi}\right| \cr
&+ \left|J_{\phi}^\top\right| \left|H_{\phi'}\right| \left|J_{\phi'}-J_{\phi}\right|+\left|J_{\phi'}^\top-J_{\phi}^\top\right| \left|H_{\phi}\right| \left|J_{\phi}\right| \cr
&\leq b^5(p+7)^3 L^3 a^{5L-1} W^{5L-5} 2^{2L+4}\left|\phi-\phi'\right|.
\end{align*}
Thus, we have
\begin{align}\label{eq:lemma_lip_1}
\left|\mathcal{R}[f_{\phi}]-\mathcal{R}[f_{\phi'}]\right| &\leq \left(b(p+7)^3 L a^{2L-1} W^{2L-2} 2^{L}\right)\cr
&\times \left(1+4 a^L + \eta b^4 (p+7)^2 L a^{3L} W^{3L-3} 2^{L+3} \right)\left|\phi - \phi'\right|.
\end{align}
Besides, we can set $\phi'=0$ to bound $2 \max_\phi \left|\mathcal{R}[f_{\phi}]\right|$ in Eq.~(\ref{eqn:lemmaD4-1}, i.e.,
\begin{align}\label{eq:lemma_lip_2}
2 \max_\phi \left|\mathcal{R}[f_{\phi}]\right| &\leq \left(b(p+7)^3 L a^{2L-1} W^{2L-2} 2^{L}\right)\cr
&\times \left(a+4a^{L+1} + \eta b^4 (p+7)^2 L a^{3L+1} W^{3L-3} 2^{L+3} \right).
\end{align}
Combine Eq.~\eqref{eq:lemma_lip_1} and Eq.~\eqref{eq:lemma_lip_2}, we have 
\begin{align}
\mathcal{L} &\leq \left(b^2(p+7)^6 L^2 a^{4L-1} W^{4L-4} 2^{2L}\right)\left(1+4 a^L+ \eta b^4 (p+7)^2 L a^{3L} W^{3L-3} 2^{L+3} \right)^2 \cr
&=\mathcal{O}\left(\text{poly}(b,p,L,\eta,a^L,W^L)\right).
\end{align}
\end{proof}
\subsection{Generalization error analysis}\label{sec:ge}
We now present the theoretical analysis of the generalization performance of the PINN model on learning the training dynamics of VQAs. We first start with the following general setting, let $D\subset \mathbb{R}^d$ be a compact space and $u:D\rightarrow \mathbb{R}$ be the true solution for the training dynamics and $u_{\phi}:D \rightarrow \mathbb{R}$ be the PINNs approximation with parameters $W\in \mathbb{R}^d$. Let $\mathcal{S}=\{x_i\}_{i=1}^N$ be the independently sampled training data-set with probability measure $\mu$ over $D$. Here, we define the empirical risk $\mathcal{E}_T$ trained over $\mathcal{S}$ and expected risk $\mathcal{E}_E$ perspectively,
\begin{align*}
\mathcal{E}_T&=\frac{1}{\left|\mathcal{S}\right|}\sum_{j=1}^{\left|\mathcal{S}\right|}|u(x_j)-u_\phi(x_j)|^2,\cr
\mathcal{E}_E&=\int_D d\mu |u-u_\phi|^2.
\end{align*}
Here, we denote $\phi^*= \argmin_{\phi \in \mathbb{R}^m}\mathcal{E}_T$ as the optimal parameters of PINN over training set $\mathcal{S}$, then the generalization error can be decomposed as follows \citep{de2022error},
\begin{align}\label{eq:error_decomp}
\mathcal{E}_E(\phi^*)&\leq 
\sup_{\hat{\phi}\in \mathbb{R}^m}\left|\mathcal{E}_E(\phi^*)-\mathcal{E}_E(\hat{\phi})\right| +\sup_{\hat{\phi}\in \mathbb{R}^m}\left|\mathcal{E}_T(\hat{\phi})-\mathcal{E}_T(\phi^*)\right|\cr
&+ \sup_{\hat{\phi} \in \mathbb{R}^m}\left|\mathcal{E}_E(\hat{\phi})-\mathcal{E}_T(\hat{\phi})\right|+\mathcal{E}_T(\phi^*).
\end{align}
Based on this, we can utilize Hoeffding's inequality and the covering number to give an upper bound on the generalization error of PINN on learning VQAs' training dynamics.
\begin{corollary}
\label{thm:ge_pinn}
Let $L, W, p, m\in\mathbb{N}, c, k, \epsilon,\gamma, \eta>0$, and $\phi \in \left[-a,a\right]^m$ be the parameters of a tanh neural network with most $W$ width, $L$ hidden layers and activation function $\sigma$. Let $\mathcal{L}$ Lipschitz continuous of $\mathcal{E}_{E}$ and $\mathcal{E}_{T}$. The generalization error of PINN, that is trained over $\mathcal{S}=\{[(t_j,\bm{\theta}^{(j)}),(\mathcal{E}^{(j)},\bm{\theta}^{(j)})\}_{j=1}^\tau$, where $t_j$ and $\mathcal{E}^{(j)}$ are the time variable and loss vale at step $j$, respectively, $\bm{\theta}^{(j)} \in[-b,b]^p$ for approximating the training dynamics of VQAs with a fixed $\eta$ learning rate, with probability at least $1-\gamma$,
\begin{equation}
\mathcal{E}_E(\phi^*) - \mathcal{E}_T(\phi^*) \leq \sqrt{\frac{4c^2}{N}pLW^2\left(\ln\left(\frac{2a\mathcal{L}}{\epsilon}\right)+\ln\left(\frac{1}{\gamma}\right)\right)}.
\end{equation}
where $\mathcal{L}=\mathcal{O}\left(\text{poly}(b,p,L,\eta,a^L,W^L)\right)$, 
\end{corollary}
According to the Corollary~\ref{thm:ge_pinn}, when we assume the training error $\mathcal{E}_T$ is small, the generalization error $\mathcal{E}_E$ for learning the training dynamics of VQAs can be bounded by a function which scales at $\mathcal{O}(\text{poly}\left(N,W,L,p\right))$. Besides, we also notice that the data size $N$ polynomially depends on the dimension of data $p$ to guarantee a small generalization error, which overcomes the curse of dimensionality and is also found in  \citep{de2022error}.
\begin{proof}
The main proof idea follows Corollary 1 of \citep{de2022error}. First, for arbitrary $\epsilon > 0$, assume  $\mathcal{E}_E(\phi)$ and $\mathcal{E}_T(\phi)$ are $\mathcal{L}$-lipschitz, we have $\{\phi_i\}_{i=1}^{\mathcal{N}}$ to cover the parameter space $\Phi$ with balls of radius $\delta$. Thus, we can bound the first two terms of Eq.~\eqref{eq:error_decomp},
\begin{align}
&\sup_{\hat{\phi}\in \mathbb{R}^m:|\phi-\hat{\phi}|\leq \delta}\left|\mathcal{E}_E(\hat{\phi})-\mathcal{E}_E(\phi^*)\right| +\sup_{\hat{\phi}\in \mathbb{R}^m:|\phi-\hat{\phi}|\leq \delta}\left|\mathcal{E}_T(\hat{\phi})-\mathcal{E}_T(\phi^*)\right| \\
& \leq \sup_{\hat{\phi}\in \mathbb{R}^m}\left|\mathcal{E}_E(\hat{\phi})-\mathcal{E}_E(\phi^*)\right| +\sup_{\hat{\phi}\in \mathbb{R}^m}\left|\mathcal{E}_T(\hat{\phi})-\mathcal{E}_T(\phi^*)\right|,
\end{align}
where 
\begin{equation}\label{eq:err_dcomp_1}
\sup_{\hat{\phi}\in \mathbb{R}^m:|\phi-\hat{\phi}|\leq \delta}\left|\mathcal{E}_E(\phi^*)-\mathcal{E}_E(\hat{\phi})\right| +\sup_{\hat{\phi}\in \mathbb{R}^m:|\phi-\hat{\phi}|\leq \delta}\left|\mathcal{E}_T(\hat{\phi})-\mathcal{E}_T(\phi^*)\right|  \leq \epsilon.
\end{equation}
Besides, as parameter space $\Phi$ is compact and $\delta$-covered by $\{\phi_i\}_{i=1}^{\mathcal{N}}$, thus for any $\phi_i, i\in[\mathcal{N}]$ we also have
\begin{align}\label{eq:err_dcomp_2}
\mathcal{E}_E(\phi^*)\leq \left|\mathcal{E}_E(\phi^*)-\mathcal{E}_E(\phi_i)\right| + \left|\mathcal{E}_T(\phi^*) - \mathcal{E}_T(\phi_i)\right| + \left|\mathcal{E}_E(\phi_i)-\mathcal{E}_T(\phi_i)\right|+\mathcal{E}_T(\phi^*).
\end{align}
As we can define a projection function $f_P$ that maps $\phi$ to its nearest cover center $\phi_i$, $f_P$ partition the parameter space $\Phi$ into $\mathcal{N}$ regions and $\forall \phi \in \Phi, \sum_i \mathcal{P}(f_P(\phi)=\phi_i)=1$. As $\mathcal{E}_E(\phi)=\mathbb{E}\left[\mathcal{E}_T(\phi)\right]$, we first employ the Hoeffding's equation to get 
\begin{equation}
\mathcal{P}(\mathcal{E}_E(\phi_j)-\mathcal{E}_T(\phi_j)\leq \epsilon | j\in[\mathcal{N}])\geq 1- \exp\left(\frac{-\epsilon^2 N}{2c^2}\right).
\end{equation}
Then, let the radius be $\delta=\epsilon/2\mathcal{L}$, then the covering number $\mathcal{N}$ can be bounded by $\left(2a\mathcal{L}/\epsilon\right)^m$. As such, we take a union bound over $\mathcal{N}$ and achieve
\begin{equation}
\mathcal{P}(\exists \phi_j, \mathcal{E}_E(\phi_j)-\mathcal{E}_T(\phi_j) \leq \epsilon) \geq 1- \left(\frac{2a\mathcal{L}}{\epsilon}\right)^m\exp\left(\frac{-\epsilon^2N}{2c^2}\right).
\end{equation}
and 
\begin{equation}
\mathcal{P}(\mathcal{E}_E(\phi^*)-\mathcal{E}_T(\phi^*) \leq \epsilon) \geq \mathcal{P}(\exists \phi_j, f_{P}(\phi^*)=\phi_j, \mathcal{E}_E(\phi_j)-\mathcal{E}_T(\phi_j) \leq \epsilon).
\end{equation} 
Thus, by combining them, we have
\begin{align}
\mathcal{P}(\mathcal{E}_E(\phi^*)-\mathcal{E}_T(\phi^*)\leq \epsilon) \geq 1- \left(\frac{2a\mathcal{L}}{\epsilon}\right)^m\exp\left(\frac{-\epsilon^2N}{2c^2}\right).
\end{align}
Therefore, we have a generalization error bound, with probability at least $1-\gamma$ as follows
\begin{equation}
\mathcal{E}_E(\phi^*) - \mathcal{E}_T(\phi^*) \leq \sqrt{\frac{2c^2}{N}m\left(\ln\left(\frac{2a\mathcal{L}}{\epsilon}\right)+\ln\left(\frac{1}{\gamma}\right)\right)}.
\end{equation}
If PINN is constructed using an $L$-layer tanh neural network with most $W$ width of each layer, it has most $(L-2)W^2+(p+1)W$ weights and $(L-1)W+1$ biases. Consequently, $m \leq 2pLW^2$. Then, using the Lemma \ref{lemma:lipsitch_E}, i.e. $\mathcal{L} \leftarrow \mathcal{O}\left(\text{poly}(b,p,L,a,W)\right)$, we have
\begin{equation}
\mathcal{E}_E(\phi^*) - \mathcal{E}_T(\phi^*) \leq \sqrt{\frac{4c^2}{N}pLW^2\left(\ln\left(\frac{2a\mathcal{O}\left(\text{poly}(b,p,L,\eta,a^L,W^L)\right)}{\epsilon}\right)+\ln\left(\frac{1}{\gamma}\right)\right)}.
\end{equation}
\end{proof}
 
\section{Details of Experiments}\label{sec:appendix details of experiments}
\subsection{Computational resources for all experiments}
Most of the simulations were run on Dual NVIDIA GeForce RTX 4090 GPUs with a 96-core AMD EPYC 9654 Processor and 256 GiB of memory.

\subsection{Variational quantum ansatz in VQE}\label{sec:appendix details of ansatz}

\paragraph{Hardware-Efficient Ansatz (HEA)}
HEAs are a class of variational quantum circuits whose structure is primarily dictated by the connectivity and native gate operations available on a specific quantum computing hardware platform. The HEA typically consists of a repetitive structure of single-qubit rotation gates and fixed entangled gates that can be implemented directly and efficiently on the target hardware, often without requiring complex gate decompositions or extensive qubit routing \cite{leone2024practical}. Concretely, it can be expressed as  
\begin{equation}  
U_{\text{HEA}}(\boldsymbol{\theta}) = \prod_{l=1}^{L} \left( \prod_{i=1}^{n} R_{i,l} (\theta_{i,l}) \prod_{(i,j)\in E} U_{\text{ent}}^{(i,j)} \right),  
\end{equation}  
where $R_{i,l}(\theta_{i,l})$ refers to single-qubit rotation gates at $l$-th layer acting on $i$-th qubit, $U_{\text{ent}}^{(i,j)}$ is entanglement gate applied to pairs of qubits $(i,j)$ that are connected according to a predefined graph $E$ that typically reflects the physical connectivity of the qubits on the quantum hardware, ensuring that the entangling gates are applied only to directly connected qubits. In our experiment, we use $R_y$ and $R_z$ gates for single-qubit rotations and CZ gates for building the $L$-layer HEA with $2nL$ variational parameters.

\paragraph{Hamiltonian Variational Ansatz (HVA)} HVA is a class of parameterized quantum circuits, the structure of which is inspired by the time evolution operator under the given Hamiltonian $H=\sum_k H_k$, often constructed as a sequence of exponential terms in the Hamiltonian \cite{park2024hamiltonian}. By parameterizing the evolution time or related coefficients, the HVA explores the quantum state space in a way that is naturally aligned with the system dynamics, potentially leading to efficient encoding of low-energy states. Generally, it can be written as
\begin{equation}
U_{\text{HVA}}(\bm{\theta}) = \prod_{l=1}^{L} \left( \prod_{k=1}^{K} e^{-i \theta_{k,l} H_k} \right),
\end{equation}
if $H_k$ is Pauli strings, each evolution operator $e^{-i \theta_{k,l} H_k}$ can be implemented using a sequence of $\{\text{H}, \text{S}, \text{S}^\dagger, \text{CNOT}, R_z\}$. For instance, if $H_k = XYZ$, the circuit implementation of $e^{-i X\otimes Y\otimes Z}$as shown in Fig.~\ref{fig:exp xyz}. The number of layers $L$ controls the expressivity of the ansatz. This form directly incorporates the structure of the problem's Hamiltonian into the design of the variational circuit.
\begin{figure}[!h]
  \centering
\includegraphics[width=0.6\columnwidth]{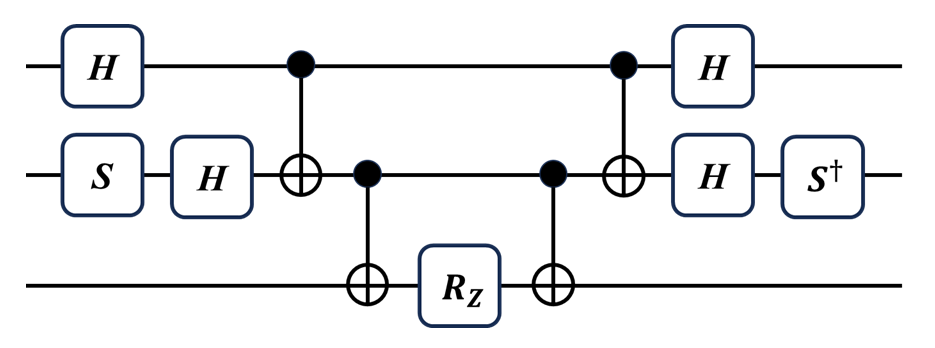}
  \caption{The circuit implementation of $e^{-i X\otimes Y\otimes Z}$.}
  \label{fig:exp xyz}
\end{figure}

\paragraph{Unitary Coupled-Cluster Singles and Doubles (UCCSD) Ansatz}
The UCCSD ansatz is a chemistry-inspired variational quantum circuit widely used in quantum computational chemistry \cite{barkoutsos2018quantum,guo2024experimental,boy2025energy}. The electronic structure Hamiltonian in quantum chemistry is expressed in second quantization as  
\begin{equation}  
H = \sum_{pq} h_{p,q} \hat{a}_p^\dagger \hat{a}_q + \frac{1}{2} \sum_{p,q,r,s} h_{pqrs} \hat{a}_p^\dagger \hat{a}_q^\dagger \hat{a}_r \hat{a}_s,  
\end{equation}  
where $\hat{a}_p^\dagger$ and $\hat{a}_q$ are fermionic creation and annihilation operators, and $h_{pq}$, $h_{pqrs}$ represent one- and two-electron integrals encoding kinetic energy, nuclear attraction, and electron-electron repulsion. The variational wavefunction is given by  
\begin{equation}  
|\Psi(\boldsymbol{\theta})\rangle = e^{T - T^\dagger} | \Phi_0 \rangle,  
\end{equation}  
where $|\Phi_0\rangle$ is the Hartree-Fock state, and $T = T_1 + T_2$ consists of single and double excitation operators:  
\begin{equation}  
T_1 = \sum_{i,m} \theta_{i}^m \hat{a}_m^\dagger \hat{a}_i, \quad  
T_2 = \sum_{i,j,m,n} \theta_{i,j}^{m,n} \hat{a}_n^\dagger \hat{a}_m^\dagger \hat{a}_j \hat{a}_i.  
\end{equation}  
Here, $i,j$ index occupied orbitals, $m,n$ index virtual orbitals, and $\boldsymbol{\theta}$ denotes variational parameters. The Jordan-Wigner transformation maps fermionic operators $\hat{a}$ and $\hat{a}^\dagger$ onto qubit operators, ensuring preservation of anticommutation relations and enabling implementation on quantum hardware.  
In our experiment, we use the BeH$_2$ molecule as an example. The mapped Hamiltonian requires 14 qubits, and the UCCSD ansatz involves 90 variational parameters.

\subsection{Details of benchmarks}
\label{sec: appendix details of benchmarks}
\paragraph{Long-short Term Memory (LSTM)}
The LSTM model employed in our study adopts a standard recurrent architecture, specifically tailored for sequence modeling and parameter optimization tasks. It consists of an LSTM layer with one hidden layer, where the input size corresponds to the number of variational parameters $p$, and the hidden size is set to $50 \times p$ to enhance its representational capacity. The model takes as input a sequence of past optimization states with a predefined sequence length $\tau_{\text{LSTM}}$, allowing it to learn temporal dependencies in parameter evolution. It processes input sequences in a batch-first manner to ensure efficient training. The final hidden state of the LSTM, corresponding to the last time step, is passed through a fully connected linear layer to produce the output, which has the same dimensionality as the input parameters. This structure enables the model to leverage past optimization information effectively to enhance parameter updates. 

\textbf{QuACK.}
For the QuACK model, we adopt the specific implementation of Dynamic Mode Decomposition (DMD) as proposed in \cite{luo2024QuACK}. This approach leverages the Koopman operator learning algorithm to find an appropriate embedding space where the system dynamics can be approximated as linear. By mapping the variational parameter updates into this learned representation, QuACK enables more efficient optimization within the VQA framework. 
In our implementation, we define the number of samples used per training iteration as $\tau_{\text{QuACK}}$, which determines the number of past optimization steps considered for learning the underlying dynamical structure. This parameter plays a crucial role in capturing the temporal evolution of variational parameters while ensuring the stability and generalization ability of the learned model.

\subsection{Details of experimental setup}\label{sec:appx_details_of_expm}

\paragraph{Estimation of Shot Numbers for Measurement}
We now estimate the measurement on a real quantum computer. The estimation strategy follows the approach outlined in \cite{peruzzo5variational}, where the number of the Pauli strings in a Hamiltonian is denoted by $M$, and the target accuracy for the expected value of the measurement is $\varepsilon$. The required number of shots for measuring the expected value of the Hamiltonian is $\mathcal{O}(M / \varepsilon^2)$. Therefore, the required number of shots for one VQE iteration, given $p$ as the number of variational parameters, can be estimated as $2 \times p \times M / \varepsilon^2$. We use $\varepsilon = 1 \times 10^{-3}$ in our specific calculation.

\paragraph{Performance on Different Ansatz} 
In the experimental setup of the 12-qubit TFIM, the 3-layer HEA has a total of $2 \times 12 \times 3$ variational parameters. For the 14 qubits $\text{BeH}_2$ system, the USCCSD ansatz involves 90 variational parameters. In both experiments, the network architecture and training procedure of PALQO follow the standard settings described in Appendix~\ref{sec:implementation_detail}, with a maximum training epoch of $T_{\text{epoch}} = 3400$ and $\tau=2$ training samples per cycle. The maximum number of LSTM training iterations is $T_{\text{epoch}} = 2000$, with $\tau_{\text{LSTM}}=3$ training samples per cycle. Additionally, the number of samples $\tau_{\text{QuACK}}=3$ is used by QuACK.

\textbf{Scalability}
In this experiment, the number of variational parameters in an $n$-qubit TFIM with an $L$-layer HEA is given by $p = 2 \times n \times L$, where $L=\{2,3,4,5,6,7,8\}$. In experiments conducted with $n=4$ to $n=40$ qubits using a fixed 3-layer HEA, the network architecture in PALQO follows the settings in Appendix~\ref{sec:implementation_detail}, with the maximum number of training epochs $T_{\text{epoch}}$ set to \{3000, 3000, 3000, 3500, 3500, 3500, 3500, 4000, 4000, 4000\}. Additionally, the number of samples used in the first cycle is set to $\tau=1$ for the 4-qubit system. For systems with sizes between 4 and 12, $\tau=2$ samples are employed in subsequent cycles, while for larger systems with sizes ranging from 16 to 40, $\tau=3$ samples are used.
In experiments with a fixed 12-qubit system and varying HEA layers from 2 to 8, the maximum number of training epochs $T_{\text{epoch}}$ follows \{3000, 3000, 3500, 3500, 3500, 3500, 4000\}. In this setting, except for the first cycle, the number of training samples used per cycle remains $\tau = 2$.

\section{Additional Numerical Experiments}\label{sec:appendix additional experiments}
In this section, we present additional numerical experiments to further validate the superior performance of PALQO. Specifically, we evaluate its effectiveness in three representative tasks: the XXZ model, the LiH molecule, and a quantum machine learning (QML) classification problem. We also examine the robustness of PALQO in the presence of quantum noise. Moreover, our results indicate that PALQO can be effectively integrated with resource-saving techniques, such as measurement grouping, to further reduce quantum resource consumption during the VQA optimization process.

\subsection{XXZ and LiH}
We present the additional numerical experiments of performance comparisons of PALQO on 12 qubits XXZ with HVA, and 14 qubits LiH with UCCSD ansatz for varying structural parameters are presented in Fig.~\ref{fig:XXZ LiH}. The results demonstrate that PALQO achieves lower $\Delta E$ and higher speed ratio in most cases. In the case of  $J=1, \delta = 0.5$, PALQO exhibits a comparable speedup ratio to the reference methods, primarily due to the smaller energy gap in this setting, which makes the optimization landscape more challenging and hinders PALQO's convergence efficiency.
\begin{figure}[!h]
\begin{center}
\centerline{\includegraphics[width=1\columnwidth]{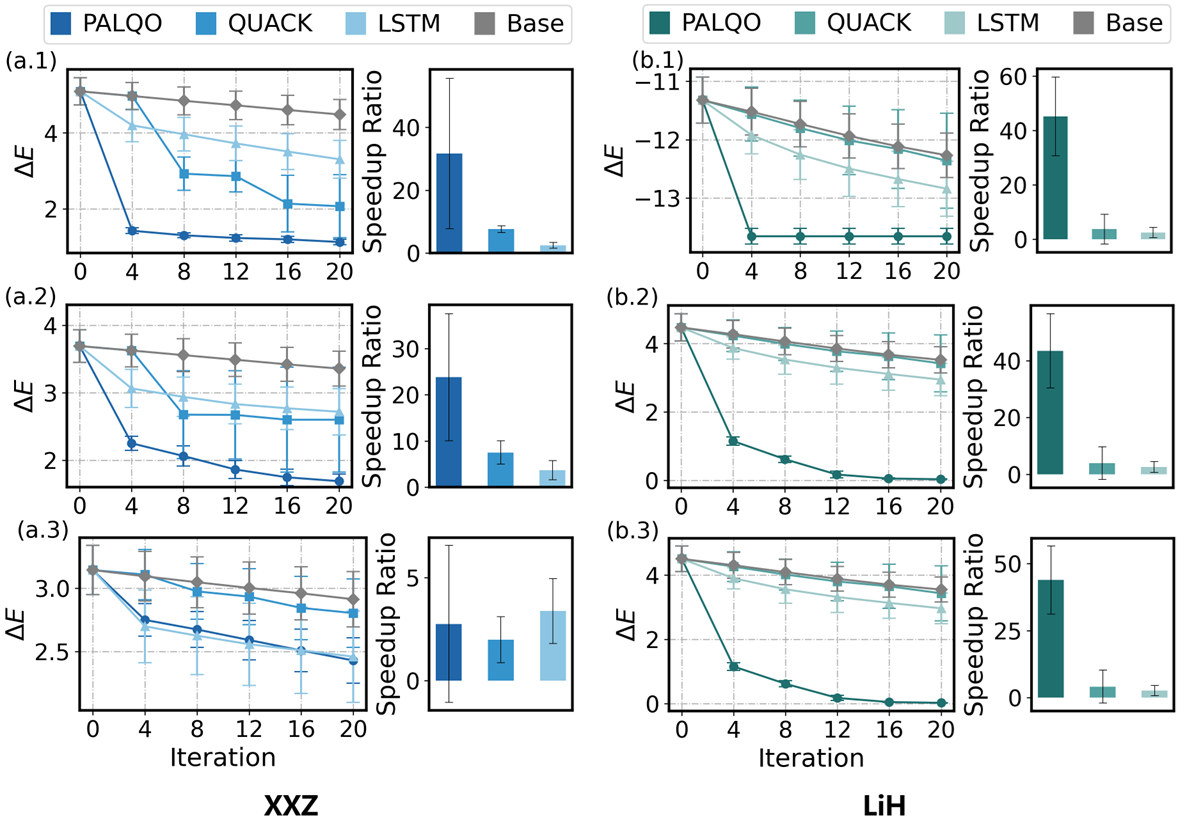}}
\caption{\small Performance comparison between PALQO and the reference models in XXZ with HVA and 12-qubit LiH with UCCSD ansatz. Each subplot comprises a $\Delta E$ curve over iterations performed on a quantum device, along with a bar chart depicting the speedup ratios achieved by PALQO and competing models. The left column illustrates results for XXZ with $J=J'=1,\delta=\{2,1,0.5\}$. The right column displays the model performance on LiH\(2\) with the bond length $b = \{1.4,1.5,1.6\}$. }
\label{fig:XXZ LiH}
\end{center}
\end{figure}

\subsection{Quantum machine learning}
To further assess the applicability of PALQO in other VQAs like quantum neural network (QNN), we conduct experiments on a classification task. Based on Eq.~\eqref{eq:dynamics qnn}, we rebuild PALQO for QNN with the reformulated cost function. We construct the 4-qubit QNN with 3-layer HEA and measurement observable $O = I\otimes I \otimes Z \otimes Z$ as the baseline model, and employ the quantum circuit shown in Fig.~\ref{fig:qnn encoder} as the feature encoder to map classical input data into quantum states. The performance comparison between PALQO and the baseline model on a classification task over the Iris dataset \cite{iris_53} is shown in Fig.~\ref{fig:qnn results}. In Fig.~\ref{fig:qnn results}~(a), it shows that PALQO achieves significantly lower loss values than the baseline throughout the iterations. During the initial optimization phase, PALQO is capable of more rapidly reaching the points with lower loss, which in turn reduces the optimization steps. In Fig.~\ref{fig:qnn results}~(b), as PALQO can more swiftly attain lower loss values, it reaches an average accuracy over 90\% by the 120 steps, significantly outperforming the baseline model, which only achieves 75\%. Therefore, it indicates that the robust applicability of PALQO while exhibiting favorable performance.

\begin{figure}[!h]
\begin{center}
\centerline{\includegraphics[width=.7\columnwidth]{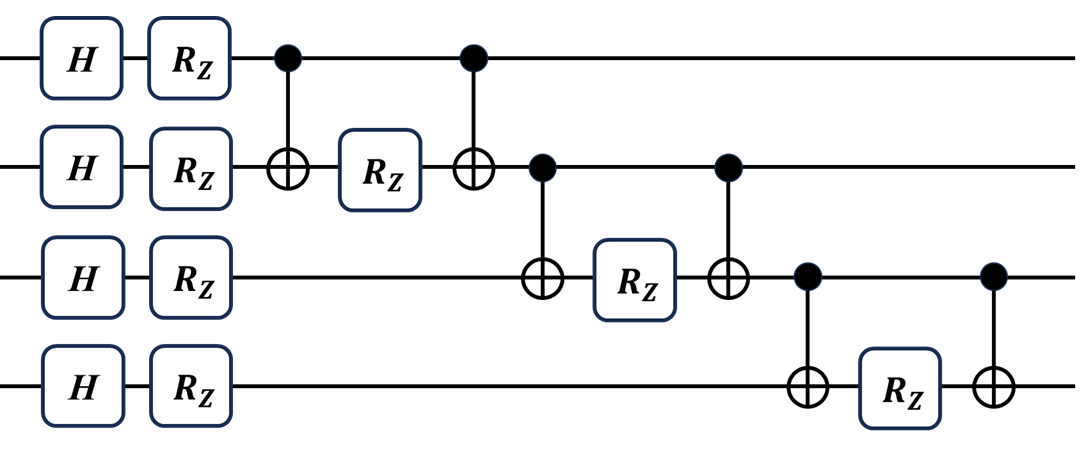}}
\caption{The illustration of a quantum encoder circuit. It employs an instantaneous quantum polynomial (IQP) encoding strategy for QNN \cite{havlivcek2019supervised}, in which data features are embedded into the rotation angles of parameterized quantum gates such as $R_x, R_z$. In our implementation, the Iris dataset features $\bm{x}=(x_0,\cdots,x_7)$  are individually encoded into the rotation angles of 7 corresponding parameterized gates.}
\label{fig:qnn encoder}
\end{center}
\end{figure}

\begin{figure}[!h]
\begin{center}
\centerline{\includegraphics[width=\columnwidth]{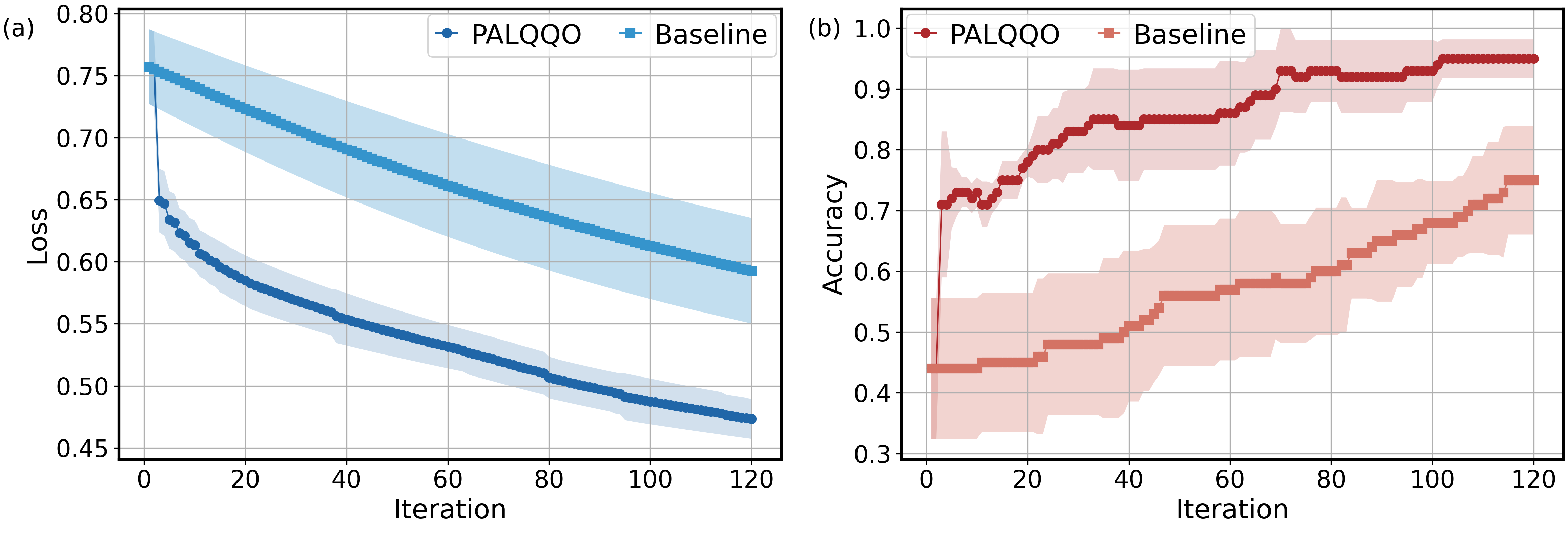}}
\caption{Performance comparison between PALQO and the baseline method on a quantum machine learning classification task using the Iris dataset. (a) The loss curve between PALQO and the baseline model. (b) The accuracy curve of PALQO versus the baseline model over the iterations. Shaded regions refer to the range of the loss and accuracy over multiple runs. }
\label{fig:qnn results}
\end{center}
\end{figure}

\subsection{Performance under noise}
\begin{figure}[!h]
\begin{center}
\centerline{\includegraphics[width=0.7\columnwidth]{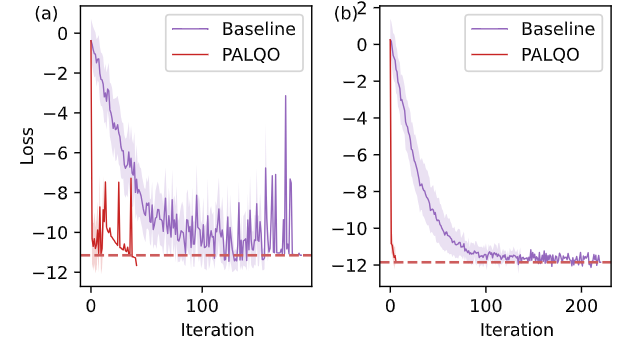}}
\caption{Performance of PALQO under noise conditions. (a) Optimization results with 5\% depolarizing noise. (b) Performance under shot noise with a shot count of 100.}
\label{Fig.4}
\end{center}
\end{figure}
We further assess the robustness of PALQO in the presence of noise, specifically evaluating its performance on a 12-qubit TFIM with a 3-layer HEA. In this experiment, we test ten randomly initialized sets of variational parameters for each noise scenario.  
The results are presented in Fig.~\ref{Fig.4}. Despite the presence of noise, PALQO consistently demonstrates strong performance, highlighting its robustness and practical applicability in realistic quantum environments.

\subsection{Complement to measurement optimization}
Here, we provide the results of the complementary experiments of PALQO and measurement grouping on 20-qubit TFIM, 12-qubit LiH, and 14-qubit BeH\(2\), which demonstrate that PALQO offers a valuable complement to existing strategies for further enhancing the optimization efficiency of VQAs. Measurement grouping strategically reduces the number of distinct measurements by exploiting the commutativity of Hamiltonian terms, thereby enabling the simultaneous measurement of multiple observables. Thus, PALQO can seamlessly incorporate measurement grouping into the overall framework. As shown in Fig.~\ref{fig:complementary}, rather than replacing grouping strategies, our method works in tandem with them, offering a multi-faceted approach to further reduce the quantum resource burden. 
\begin{figure}[!h]
\begin{center}
\centerline{\includegraphics[width=0.5\columnwidth]{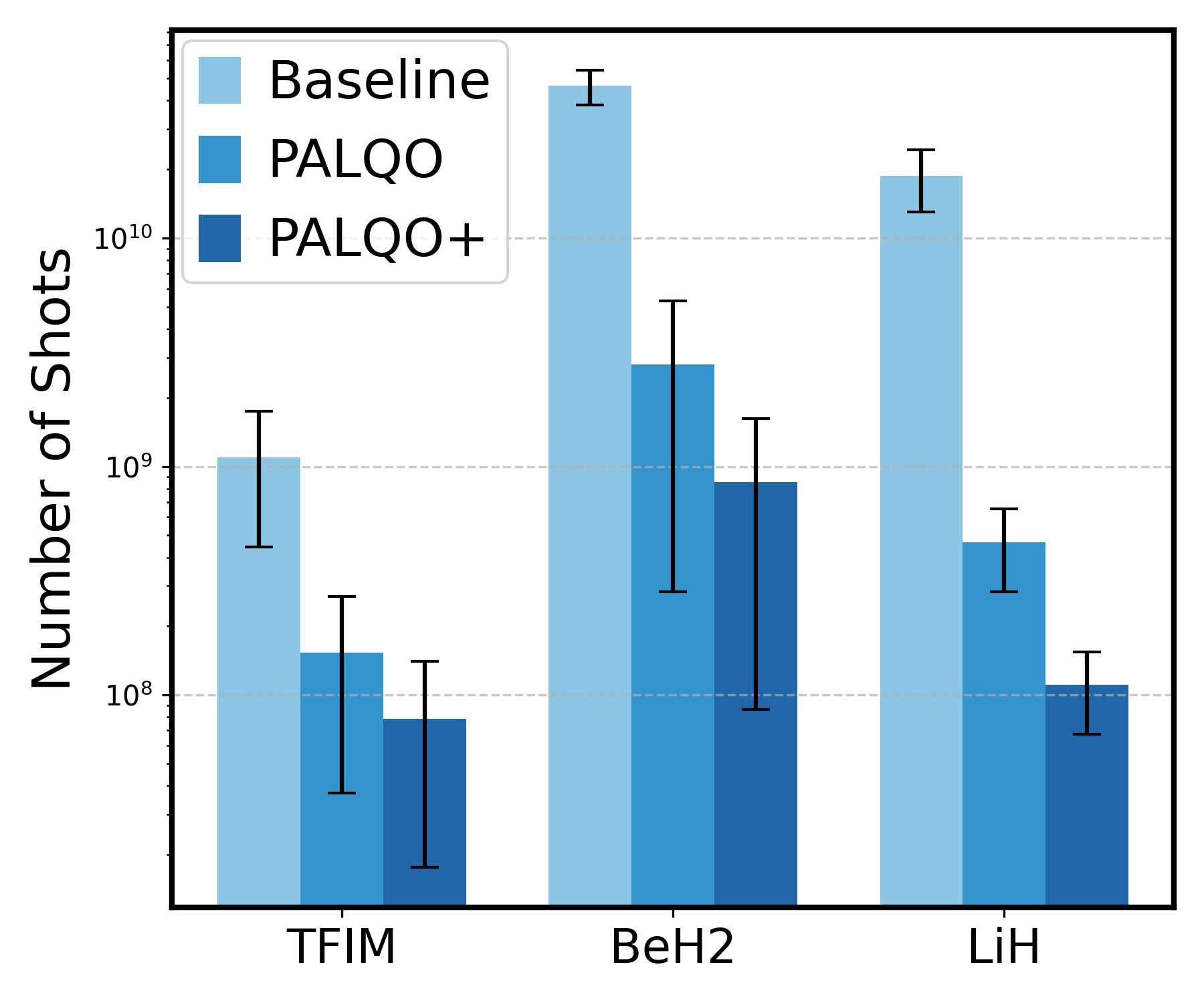}}
\caption{The measurement shots of PALQO, combined with measurement grouping, are evaluated on tasks including the TFIM, LiH, and BeH\(2\). PALQO$+$ refers to the PALQO enhanced by measurement grouping.}
\label{fig:complementary}
\end{center}
\end{figure}


\end{document}